\documentclass[11pt]{article}
\usepackage{lmodern}
\usepackage[margin=1.5in]{geometry}
\usepackage[utf8]{inputenc}
\usepackage{amsmath}
\usepackage{amssymb}
\usepackage{amsthm}
\allowdisplaybreaks

\usepackage{latexsym}
\usepackage{graphicx}
\usepackage{ifthen}
\usepackage{calc}
\usepackage{stmaryrd}
\usepackage{enumitem}
\usepackage{hyperref}

\usepackage{tikz}
\usetikzlibrary{shapes}

\makeatletter
\renewcommand{\section}{\@startsection
  {section}%
  {1}%
  {0mm}%
  {-1\baselineskip}%
  {0.5\baselineskip}%
  {\normalfont\large\bfseries}%
}
\renewcommand{\subsection}{\@startsection
  {subsection}%
  {2}%
  {0mm}%
  {-1\baselineskip}%
  {0.5\baselineskip}%
  {\normalfont\large\itshape}%
}

\renewcommand{\subsubsection}{\@startsection
  {subsubsection}%
  {3}%
  {0mm}%
  {-1\baselineskip}%
  {0.5\baselineskip}%
  {\normalfont\itshape}%
}

\newsavebox{\tempbox}
\renewcommand{\@makecaption}[2]{
  \vspace{10pt}
  \sbox{\tempbox}{\textbf{#1.} #2}
  \ifthenelse{\lengthtest{\wd\tempbox > \linewidth}}{
    \textbf{#1.} #2\par
  }{
    \begin{center}
      \textbf{#1.} #2
    \end{center}
  }
}

\makeatother

\newtheoremstyle{mythm}%
  {}%
  {}%
  {\itshape}%
  {}%
  {\bfseries}%
  {.}%
  {.5em}%
  {\thmname{#1}~\thmnumber{#2}\ifthenelse{\equal{\thmnote{#3}}{}}{}{~(\thmnote{#3})}}%

\newtheoremstyle{mydefn}%
  {}%
  {}%
  {\upshape}%
  {}%
  {\bfseries}%
  {.}%
  {.5em}%
  {\thmname{#1}~\thmnumber{#2}\ifthenelse{\equal{\thmnote{#3}}{}}{}{~(\thmnote{#3})}}%

\newtheoremstyle{myremark}%
  {}%
  {}%
  {\upshape}%
  {}%
  {\itshape}%
  {.}%
  {.5em}%
  {\thmname{#1}~\thmnumber{#2}\ifthenelse{\equal{\thmnote{#3}}{}}{}{~(\thmnote{#3})}}%

\theoremstyle{mythm} \newtheorem{theorem}{Theorem}[section]
\newtheorem{lemma}[theorem]{Lemma}

\newtheorem{corollary}[theorem]{Corollary}

 \theoremstyle{mydefn}

\newtheorem{example}[theorem]{Example} \theoremstyle{myremark}
\newtheorem{remark}[theorem]{Remark} \theoremstyle{mythm}
                      
\newcounter{claimcounter}
\newenvironment{claim}[1][]{
  \renewcommand{\proof}{\smallskip\par\noindent\textit{Proof. }}
  \medskip\par\noindent%
  \ifthenelse{\equal{#1}{}}{%
    \setcounter{claimcounter}{0}\refstepcounter{claimcounter}\textit{Claim~\arabic{claimcounter}.}
  }{%
    \ifthenelse{\equal{#1}{resume}}{%
      \refstepcounter{claimcounter}\textit{Claim~\arabic{claimcounter}.}
    }{%
      \textit{Claim~#1.}
    }
  }
}{
  \par\medskip
}

\newcommand{\uenda}{\tag*{$\lrcorner$}}

\numberwithin{equation}{section}

\newcommand{\Q}{\mathsf{Q}}

\newcommand{\FF}{\mathbb F}

\newcommand{\Sol}{\alpha}

\newlist{caselist}{description}{10}
\setlist[caselist]{font=\itshape\mdseries}

\setenumerate[1]{label=(\arabic*)}
\newlist{eroman}{enumerate}{2}
\setlist[eroman,1]{label=(\roman*)}
\setlist[eroman,2]{label=(\alph*)}
\newlist{ealph}{enumerate}{1}
\setlist[ealph]{label=(\Alph*)}

\newcounter{nlistcounter}

\renewcommand{\phi}{\varphi}
\newcommand{\bigmid}{\;\big|\;}
\newcommand{\Bigmid}{\;\Big|\;}

\renewcommand{\mathbf}[1]{\textit{\bfseries #1}}

\usepackage{color}
\definecolor{gruen}{rgb}{0,0.6,0.2}

\newcounter{rbcounter}
\renewcommand{\tilde}[1]{\widetilde{#1}}
\newcommand{\CC}{{\mathcal C}}

\newcommand{\CY}{{\mathcal Y}}
\newcommand{\CX}{{\mathcal X}}

\newcommand{\TS}{{\mathcal{TS}}}

\newcommand{\bzeta}{{\boldsymbol\zeta}}
\newcommand{\bseta}{{\boldsymbol\eta}}

\newcommand{\SP}{{\mathsf P}}
\newcommand{\SQ}{{\mathsf Q}}
\newcommand{\SR}{{\mathsf R}}
\newcommand{\Piso}{{\mathsf P}_{\textup{iso}}}
\newcommand{\GIeq}[1]{\Piso(#1)}
\newcommand{\PTs}{{\mathsf P}_{\textup{Ts}}}

\newcommand{\tp}{\operatorname{tp}}

\begin{document} 
\title{Limitations of Algebraic Approaches to Graph Isomorphism Testing}

\author{Christoph Berkholz and 
Martin Grohe\\\normalsize RWTH Aachen University\\\normalsize\{grohe,berkholz\}@informatik.rwth-aachen.de}
\date{}

\maketitle

\begin{abstract}
  We investigate the power of graph isomorphism algorithms based on algebraic
  reasoning techniques like Gr\"obner basis computation. The idea of
  these algorithms is to encode two graphs into a system of equations
  that are satisfiable if and only if if the graphs are isomorphic,
  and then to (try to) decide satisfiability of the system using, for
  example, the Gröbner basis algorithm.  In some cases this can be
  done in polynomial time, in particular, if the equations admit a
  bounded degree refutation in an algebraic proof systems such as
  Nullstellensatz or polynomial calculus. We prove linear lower bounds
  on the polynomial calculus degree over all fields of characteristic
  $\neq2$ and also linear lower bounds for the degree of
  Positivstellensatz calculus derivations.

  We compare this approach to recently studied linear and semidefinite
  programming approaches to isomorphism testing, which are known to be
  related to
  the combinatorial Weisfeiler-Lehman algorithm. We exactly
  characterise the power of the Weisfeiler-Lehman algorithm 
  in terms of an algebraic proof system that lies between
  degree-$k$ Nullstellensatz and degree-$k$ polynomial calculus.
\end{abstract}

\section{Introduction}
The graph isomorphism problem (GI) is notorious for its unresolved
complexity status. While there are good reasons to believe that GI is
not NP-complete, it is wide open whether it is in polynomial time.

Complementing recent research on linear and semidefinite programming
approaches to GI
\cite{atsman13,codschsno14,Grohe.2012,mal14,DWWZ.2013}, we investigate
the power of GI-algorithms based on algebraic reasoning techniques
like Gr\"obner basis computation. The idea of all these approaches is
to encode isomorphisms between two graphs as solutions to a system of
equations and possibly inequalities and then try to solve this system
or relaxations of it. Most previous work is based on the following
encoding: let $G,H$ be graphs with adjacency matrices $A,B$,
respectively. Note that $G$ and $H$ are isomorphic if and only if
there is a permutation matrix $X$ such that $AX=XB$.  If we view the
entries $x_{vw}$ of the matrix $X$ as variables, we obtain a system of
linear equations. We introduce equations forcing all row- and column sums of
$X$ to
be $1$ and add the inequalities $x_{vw}\ge0$. It follows that the integer solutions to
this system are 0/1-solutions that correspond to isomorphisms between $G$ and $H$. Of course
this does not help to solve GI, because we cannot find integer
solutions to a system of linear inequalities in polynomial time. The
first question to ask is what happens if we drop the integrality
constraints. Almost thirty years ago, Tinhofer~\cite{tin86} proved
that the system has a rational (or, equivalently, real) solution if
and only if the so-called colour refinement algorithm does not
distinguish the two graphs. \emph{Colour refinement} is a simple
combinatorial algorithm that iteratively colours the vertices of a
graph according to their ``iterated degree sequences'', and, to
distinguish two graphs, tries to detect a difference in their colour
patterns.
For every $k$, there is a natural generalisation of the colour
refinement algorithm that colours $k$-tuples of vertices instead of
single vertices; this generalisation is known as the
\emph{$k$-dimensional Weisfeiler-Lehman algorithm ($k$-WL)}. Atserias
and Maneva~\cite{atsman13} and independently Malkin~\cite{mal14}
proved that the Weisfeiler-Lehman algorithm is closely tied to
the \emph{Sherali-Adams} hierarchy \cite{Sherali.1990} of increasingly tighter
LP-relaxations of the integer linear program for GI described above:
the distinguishing power of $k$-WL is between that of the $(k-1)$st
and $k$th level of the Sherali-Adams hierarchy.
Otto and the second author of this paper~\cite{Grohe.2012} gave a
precise correspondence between $k$-WL and the nonnegative solutions to
a system of linear equations between the $(k-1)$st and $k$th level of
the Sherali-Adams hierarchy. Already in 1992, Cai, F{\"u}rer, and
Immerman~\cite{caifurimm92} had proved that for every $k$ there are
non-isomorphic graphs $G_k,H_k$ (called \emph{CFI-graphs} in the
following) of size $O(k)$ that are not distinguished by $k$-WL, and
combined with the results of Atserias-Maneva and Malkin, this implies
that no sublinear level of the Sherali-Adams hierarchy suffices to
decide isomorphism.  O'Donnell, Wright, Wu, and Zhou~\cite{DWWZ.2013}
and Codenotti, Schoenbeck, and Snook \cite{codschsno14} studied the
Lasserre hierarchy \cite{Lasserre.2001} of semi-definite relaxations
of the integer linear program for GI. They proved that the same
CFI-graphs cannot even be distinguished by sublinear levels of the
Lasserre hierarchy.
 
However, there is a different way of relaxing the integer linear
program to obtain a system that can be solved in polynomial time: we
can drop the nonnegativity constraints, which are the only
inequalities in the system. Then we end up with a system of linear
equalities, and we can ask whether it is solvable over some finite
field or over the integers. As this can be decided in polynomial time,
it gives us a new polynomial time algorithm for graph isomorphism: we
solve the system of equations associated with the given graphs. If
there is no solution, then the graphs are nonisomorphic. (We say that
the system of equations \emph{distinguishes} the graphs.) If there is a
solution, though, we do not know if the graphs are isomorphic or
not. Hence the algorithm is ``sound'', but not necessarily
``complete''. Actually, it is not obvious that the algorithm is not
complete. If we interpret the linear equations over $\mathbb F_2$ or
over the integers, the system does distinguish the CFI-graphs (which
is not very surprising because these graphs encode systems of linear
equations over $\mathbb F_2$). Thus the lower bound techniques applied
in all previous results do not apply here. However, we construct
nonisomorphic graphs that cannot be distinguished by this system (see
Theorem~\ref{theo:int}).

In the same way, we can drop the nonnegativity constraints from the
levels of the Sherali-Adams hierarchy and then study solvability over
finite fields or over the integers, which gives us increasingly
stronger systems. Even more powerful algorithms can be obtained by
applying algebraic techniques based on Gr\"obner basis computations to
these systems. Proof complexity gives us a good framework for proving
lower bounds for such algorithms. There are algebraic proof systems
such as the polynomial calculus \cite{Clegg.1996} and the weaker
Nullstellensatz system \cite{Beame.1994} that characterise the power these
algorithms. 
The degree of refutations in the algebraic systems roughly
corresponds to the levels of the Sherali-Adams and Lasserre hierarchies
for linear and semi-definite programming,
and to the dimension of the Weisfeiler-Lehman algorithm. We identify
a fragment of the polynomial calculus, called the monomial polynomial
calculus, such that degree-$k$ refutations in this system precisely
characterise distinguishability by $k$-WL (see
Theorem~\ref{thm:monomial}).

As our main lower bounds, we prove that for every field $\FF$ of
characteristic $\neq 2$, there is a family of nonisomorphic graphs
$G_k,H_k$ of size $O(k)$ that cannot be distinguished by the
polynomial calculus in degree $k$. 
Furthermore, we prove that there is a family of nonisomorphic graphs
$G_k,H_k$ of size $O(k)$ that cannot be distinguished by the
Positivstellensatz calculus in degree $k$. 
The Positivstellensatz calculus \cite{Grigoriev.2001} is an extension of the polynomial calculus over the reals and subsumes semi-definite programming hierarchies.
Thus, our results slightly generalise
the results of O'Donnell et al.~\cite{DWWZ.2013} on the Lasserre
hierarchy (described above). Technically, our contribution is a
low-degree reduction from systems
of equations describing so-called Tseitin tautologies to the systems for graph isomorphism. Then we apply
known lower bounds \cite{Buss.1998,Grigoriev.2001} for Tseitin tautologies.

\section{Algebraic Proof systems}

\emph{Polynomial calculus (PC)} is a proof system to prove that a given system of (multivariate)
polynomial equations $\SP$  over a field $\FF$ has no
0/1-solution. We always normalise polynomial equations to the form
$p=0$ and just write $p$ to denote the equation $p=0$.
The derivation rules are the following (for polynomial equations $p \in\SP$, polynomials $f, g$, variables $x$ and field elements $a,b$):
$$
\frac{}{p}, \quad\frac{}{x^2-x}, \quad\frac{f}{xf}, \quad\frac{g\quad f}{ag+bf}.
$$
The \emph{axioms} of the systems are all $p\in\SP$ and
$x^2-x$ for all variables $x$.
A PC \emph{refutation} of $\SP$ is a derivation of $1$.
The polynomial calculus is sound and complete, that is, $\SP$
has a PC refutation if and only if it is unsatisfiable.
The \emph{degree} of a PC derivation is the maximal degree of every polynomial in the derivation. 
Originally, Clegg et. al. \cite{Clegg.1996} introduced the polynomial calculus to model Gröbner basis computation.
Moreover, using the Gröbner basis algorithm, it can be decided in time $n^{O(d)}$ whether a given system of polynomial equations has a PC refutation of degree $d$ (see \cite{Clegg.1996}).

We introduce the following restricted variant of the polynomial calculus. 
A \emph{monomial-PC} derivation is a PC-derivation where we require that the polynomial $f$ in the multiplication rule $\frac{f}{xf}$ is either a monomial or the product of a monomial and an axiom.

If we restrict the application of the multiplication rule even further
and require $f$ to be the product of a monomial and an axiom, we obtain the Nullstellensatz proof system \cite{Beame.1994}. 
This proof system is usually stated in the following static form.
A \emph{Nullstellensatz} refutation of a system $\SP$ 
of polynomial
equations consists of polynomials $f_p$, for $p\in\SP$, and
$g_x$, for all variables $x$, such that
\[
\sum_{p\in\SP}f_pp+\sum_{x}g_x(x^2-x)=1.
\]
The degree of a Nullstellensatz refutation is the maximum degree of
all polynomials $f_pp$.
\subsection{Low-Degree Reductions}
To compare the power of the polynomial calculus for different systems
of polynomial equations, we use \emph{low degree reductions} \cite{Buss.2001}.
Let $\SP$ and $\SR$ be two sets of polynomials in the variables
$\CX$ and $\CY$, respectively. A
\emph{degree-$(d_1,d_2)$ reduction} from $\SP$ to $\SR$ consist of the following:
\begin{itemize}
\item for each variable $y\in\CY$ a polynomial $f_y(x_1,\ldots,x_k)$
  of degree at most $d_1$ in variables $x_1,\ldots,x_k\in\CX$;
\item for each polynomial $r(y_1,\ldots,y_\ell)\in\SR$ a
  degree-$d_2$ PC derivation of 
  \[
  r\big(f_{y_1}(x_{11},\ldots,x_{1k_1}),\ldots,f_{y_\ell}(x_{\ell
    1},\ldots,x_{\ell k_\ell})\big)
  \]
  from $\SP$.
\item for each variable $y\in\CY$ a degree-$d_2$ PC derivation of 
  \[
  f_{y}(x_{1},\ldots,x_{k})^2-f_y(x_{1},\ldots,x_{k})
  \]
  from $\SP$.
\end{itemize}

\begin{lemma}[\cite{Buss.2001}]\label{lem:lowdeg}
If there is a degree-$(d_1,d_2)$ reduction from $\SP$ to $\SR$ and $\SR$ has a polynomial calculus refutation of degree $k$, then $\SP$ has a polynomial calculus refutation of degree $\max(d_2,kd_1)$. 
\end{lemma}

\subsection{Linearisation}

For a system of polynomial equations $\SP$ over variables $x_i$
let $\SP^{r}$ be the set of all polynomial equations of degree at most $r$ obtained
 by multiplying a polynomial in $\SP$ by
a monomial over the variables $x_i$.  Furthermore, for a
system of polynomial equations $\SP$ let
$\operatorname{MLIN}(\SP)$ be the the multi-linearisation of $\SP$
obtained by replacing every monomial $x_{i_1}\cdots x_{i_\ell}$ by a
variable $X_{\{i_1,\ldots,i_\ell\}}$. 
Observe that if
$\SP\cup \Q$ has a solution $\Sol$, then so does
$\operatorname{MLIN}(\SP)$ as we can set $\Sol(X_{\{i_1,\ldots,i_\ell\}}):=\Sol(x_{i_1})\cdots \Sol(x_{i_\ell})$.
The converse however does not hold since a solution $\Sol$
for $\operatorname{MLIN}(\SP)$ does not have to satisfy
$\Sol(X_{\{ab\}})=\Sol(X_{\{a\}})\Sol(X_{\{b\}})$.
The next lemma states a well-known connection between Nullstellensatz and Linear Algebra.
\begin{lemma}[\cite{Buss.1998}]\label{lem:NullPCLin}
  Let $\SP$ be a system of polynomial equations. The following
  statements are equivalent.
  \begin{eroman}
    \item $\SP$ has a degree $r$ Nullstellensatz refutation.
    \item The system of linear equations $\operatorname{MLIN}(\SP^r)$ has no solution.
  \end{eroman}
\end{lemma}
This characterisation of Nullstellensatz proofs in terms of a linear system of equations (also called \emph{design} \cite{Buss.1998}) is a useful tool for proving lower bounds on the Nullstellensatz degree.
Unfortunately, a similar characterisation for bounded degree PC is not in sight.
However, for the newly introduced system monomial-PC, which lies between Nullstellensatz and PC, we have a similar criterion for the non-existence of refutations.

\begin{lemma}\label{lem:monomialPC_MLIN}
  If $\operatorname{MLIN}(\SP^d)$ has a solution $\Sol$ that additionally satisfies 
  $$
  \Sol(X_{\pi})=0 \Longrightarrow \Sol(X_{\rho})=0, \text{ for all } \pi\subseteq \rho,
  $$
  then $\SP$ has no degree-$d$ monomial-PC derivation. 
\end{lemma}

\begin{proof}
  Let $\Sol$ be the assignment defined in the lemma and suppose for contradiction, that $\SP$ has a monomial-PC refutation of degree $d$.
  We define $s$ to be an evaluation function that maps that maps polynomials $h$ of degree at most $\leq d$ to elements in $\FF$.
  For field elements $a\in\FF$, we let $s(a):=a$.
  If $h=x_{i_1}\cdots x_{i_\ell}$ is a monomial we set
  $s(x_{i_1}\cdots x_{i_\ell}) := \Sol(X_{\{i_1,\ldots,i_\ell\}})$.
  If $h=\sum_j a_j\vec x_j$ is a polynomial we let $s(h):=\sum_j a_js(\vec x_j)$.  
  We now claim that $s(h)=0$ for every polynomial $h$ in the refutation. 
  This leads to a contradiction, as we finally derive $h=1$ and $s(1)=1$ by definition.
  We prove the claim by induction on the refutation.
  For the base case let $h$ be an axiom. 
  If $h=\sum_j a_j\vec x_j \in \SP$, let $\sum_j a_jX_j$ be the multi-linearisation of $h$.
  We have $s(h)=\sum_j a_js(\vec x_j)=\sum_j a_j\Sol(X_j)=0$, as $\Sol$ is a solution to the linearised equation.
  Furthermore, for axioms $x_i^2-x_i$ we have $s(x_i^2-x_i)=\Sol(X_{\{i\}})-\Sol(X_{\{i\}})=0$.
  The induction step for $\frac{g\quad f}{ag+bf}$ follows immediately as $s(ag+bf)=as(g)+bs(f)$.
  For the multiplication rule $\frac{f}{x_if}$ of monomial-PC there are two cases.
  First, if $f$ is the product of a monomial and an axiom, then $x_if\in \SP^d$.
  Hence, the linearisation of $x_if$ is in $\operatorname{MLIN}(\SP^d)$, and thus $s(x_if)=0$ as in the base case.
  If $f=x_{i_1}\cdots x_{i_\ell}$ is a monomial, then $\Sol(X_{\{i_1,\ldots,i_\ell\}})=s(f)=0$ by induction hypothesis. 
  By the additional requirement on the $\Sol$, it follows that $s(x_if)=\Sol(X_{\{i,i_1,\ldots,i_\ell\}})=0$.
  This finishes the proof of the lemma. 
\end{proof}

\subsection{Linear and Semidefinite-Programming Approaches}

In the previous section we have seen that degree-$d$ Nullstellensatz corresponds to solving a system of \emph{linear equations} of size $n^{O(d)}$, which can be done in time $n^{O(d)}$.
Over the reals, this approach can be strengthened by considering hierarchies of relaxations for linear and semi-definite programming.

In this setting one additionally adds linear inequalities, typically $0\leq x\leq 1$. 
In the same way as for the Nullstellensatz, one lifts this problem to higher dimensions, by multiplying the inequalities and equations with all possible monomials of bounded degree.
Afterwards, one linearises this system as above to obtain a system of \emph{linear inequalities} of size $n^{O(d)}$, which can also be solved in polynomial time using linear programming techniques.
This lift-and-project technique is called Sherali-Adams relaxation  of level $d$ \cite{Sherali.1990}.

Another even stronger relaxation is based on semidefinite programming techniques.
This techniques has different names: Positivstellensatz, Sum-of-Squares (SOS), or Lasserre Hierarchy.
Here we take the view point as a proof system, which was introduced by Grigoriev and  Vorobjov \cite{Grigoriev.2001} and directly extends the Nullstellensatz over the reals.
A degree-$d$ \emph{Positivstellensatz} refutation of a system
$\SP$ of polynomial equations consists of polynomials $f_p$, for $p\in\SP$, and
$g_x$, for all variables $x$, and in addition
polynomials $h_i$ such that
\[
\sum_{p\in\SP}f_pp+\sum_{x}g_x(x^2-x)=1+\sum_ih_i^2.
\]
The degree of a Positivstellensatz refutation is the maximum degree of all polynomials $f_pp$ and $h^2_i$.
It is important to note that Positivstellensatz refutations can be found in time $n^{O(d)}$ using semi-definite programming.
This has been independently observed by Parrilo \cite{Parrilo.2000} in the context of algebraic geometry and by Lasserre \cite{Lasserre.2001} in the context of linear optimisation.

Grigoriev and  Vorobjov \cite{Grigoriev.2001} also introduced a proof system called Positivstellensatz calculus, which extends polynomial calculus in the same way as Positivstellensatz extends Nullstellensatz.
A \emph{Positivstellensatz calculus} refutation of a system of polynomials $\SP$ is a polynomial calculus derivation over the reals of $1+\sum_ih_i^2$.
Again, the degree of such a refutation is the maximum degree of every polynomial in the derivation.

\section{Equations for Graph Isomorphism}

We find it convenient to encode isomorphism using different
equations than those from the system $AX=XB$ described in the
introduction. However, the equations $AX=XB$ can easily be derived in
our system (see Example~\ref{exa:derive_AX-XB}), and thus lower bounds for your system imply lower bounds
for the $AX=XB$-system.

Throughout this section, we fix graphs $G$ and $H$, possibly with
coloured vertices and/or edges. Isomorphisms between coloured graphs
are required to preserve the colours.
We assume that either $|V(G)|\ge 2$ or
$|V(H)|\ge 2$.  
We shall define a system
$\Piso(G,H)$ of polynomial equations that has a solution if and
only if $G$ and $H$ are isomorphic.  The equations are defined over
variables $x_{vw}, v\in V(G),w\in V(H)$. A solution to the system is
intended to describe an isomorphism $\iota$ from $G$ to $H$, where
$x_{vw}\mapsto 1$ if $\iota(v)=w$ and $x_{vw}\mapsto0$ otherwise.
The system $\GIeq{G,H}$ consists of the following linear and quadratic equations:
\begin{align}
  \sum_{v\in V(G)} x_{vw} - 1 &= 0 & &\text{for all }w \in V(H) \label{eq:cont1}\\
  \sum_{w\in V(H)} x_{vw} - 1 &= 0 & &\text{for all }v \in V(G) \label{eq:cont2} \\
x_{vw}x_{v'w'} &= 0 & &\parbox[t]{6cm}{for all $v,v'\in V(G),w,w'\in
                        V(H)$ such that $\{(v,w),(v',w')\}$ is no local isomorphism.}\label{eq:comp}
\end{align}
A \emph{local isomorphism} from $G$ to $H$ is an injective mapping $\pi$
with domain in $V(G)$ and range in $V(H)$ (often viewed as a subset of
$V(G)\times V(H)$) that preserves adjacencies,
that is $vw\in E(G)\iff \pi(v)\pi(w)\in E(H)$. If $G$ and $H$ are coloured
graphs, local isomorphisms are also required to preserve
colours. 

To enforce 0/1-assignments we add the following set $\Q$ of quadratic equalities
\begin{align}
x_{vw}^2 - x_{vw} &= 0& &\text{for all }v\in V(G),w\in V(H). \label{eq:quad}
\end{align}
 We treat these equations separately because they are axioms of the
 polynomial calculus anyway.
Observe that the equations \eqref{eq:cont1} and \eqref{eq:cont2} in
combination with \eqref{eq:quad} make
sure that every solution to the system describes a bijective mapping
from $V(G)$ to $V(H)$. The equations \eqref{eq:comp} make sure that this
bijection is an isomorphism.
Thus, for every field $\FF$, the system $\GIeq{G,H} \cup \Q$ has a solution over $\FF$ if and only $G$ and $H$ are isomorphic.

The following example shows how to derive the equations from the
system based on $AX=XB$ from $\Piso(G,H)$.

\begin{example}\label{exa:derive_AX-XB}
  Recall that $A$ and $B$ denote the adjacency matrices of the graphs
  $A,B$. Thus the equations from $AX=XB$ are
  \begin{equation}\label{eq:AX=XB}
  \sum_{v'\in N(v)}X_{v'w}-\sum_{w'\in N(w)}X_{vw'}=0
  \end{equation}
  for all $v\in V(G),w\in V(H)$. To derive \eqref{eq:AX=XB} from
  $\Piso(G,H)$, we multiply \eqref{eq:cont2} with $-X_{v'w}$ for
  $v'\in N(v)$ and obtain $x_{v'w}-\sum_{w'}x_{v'w}x_{vw'}=0$. Adding
  equations \eqref{eq:comp}, 
  $x_{v'w}x_{vw'}=0$, for all $w'\not\in N(w)$, yields the equation 
  $x_{v'w}-\sum_{w'\in N(w)}x_{v'w}x_{vw'}=0$. Adding these equations
  for all $v'\in N(v)$, we get
  \begin{equation}
    \label{eq:AX=XB-2}
    \sum_{v'\in N(v)}x_{v'w}-\sum_{v'\in N(v)}\sum_{w'\in N(w)}x_{v'w}x_{vw'}=0.
  \end{equation}
  Similarly, we can derive the equation
  \begin{equation}
    \label{eq:AX=XB-3}
    \sum_{w'\in N(w)}x_{vw'}-\sum_{w'\in N(w)}\sum_{v'\in N(v)}x_{v'w}x_{vw'}=0.
  \end{equation}
  Subtracting \eqref{eq:AX=XB-3} from \eqref{eq:AX=XB-2} yields
  \eqref{eq:AX=XB}. Note that the derivation has degree $2$.
\end{example}

\section{Weisfeiler-Lehman is located between Nullstellensatz and Polynomial Calculus}

To relate the Weisfeiler-Lehman algorithm to our proof systems, we use
the following combinatorial game. The \emph{bijective $k$-pebble game}
on graphs $G$ and $H$ is played by two players called \emph{Spoiler}
and \emph{Duplicator}. Positions of the game are sets
$\pi\subseteq V(G)\times V(H)$ of size $|\pi|\le k$. The game starts
in an initial position $\pi_0$. If $|V(G)|\neq |V(H)|$ or if $\pi_0$ is not a
local isomorphism, then Spoiler wins the
game immediately, that is, after $0$ rounds, Otherwise, the game is played in a sequence of
\emph{rounds}. Suppose the position after the $i$th round is
$\pi_i$. In the $(i+1)$st round, Spoiler chooses a subset $\pi\subseteq\pi_i$ of size
$|\pi|<k$. Then 
Duplicator chooses a bijection
$f:V(G)\to V(H)$. Then Spoiler chooses a vertex $v\in V(G)$, and
the new position is $\pi_{i+1}:=\pi\cup\{(v,f(v))\}$. If $\pi_{i+1}$ is
not a local isomorphism, then Spoiler wins the
play after
$(i+1)$ rounds. Otherwise, the game continues with the $(i+2)$nd round.
Duplicator wins the play if it lasts forever, that is, if Spoiler
does not win after finitely many rounds. 
\emph{Winning strategies} for either
player in the game are defined in the natural way.

\begin{lemma}[\cite{caifurimm92,hel96}]
  $k$-WL distinguishes $G$
  and $H$ if and only if Spoiler has a winning strategy for the
  bijective $k$-pebble game on $G,H$ with initial position $\emptyset$.
\end{lemma}

Observe that each game position $\pi=\{(v_1,w_1),\ldots,(v_\ell,w_\ell)\}$
of size $\ell$
corresponds to a multilinear monomial $\vec x_\pi=x_{v_1w_1}\ldots x_{v_\ell
  w_\ell}$ of degree $\ell$; for the empty position we let $\vec
x_\emptyset:=1$.

\begin{lemma}\label{lem:from_Spoiler_to_monomial-PC}
  Let $\mathbb F$ be a field of characteristic $0$.  If Spoiler has a
  winning strategy for the $r$-round bijective $k$-pebble game on $G$,
  $H$ with initial position $\pi_0$, then there is a degree $k$
  monomial-PC derivation of $\vec x_{\pi_0}$ from $\GIeq{G,H}$ over
  $\mathbb F$.
\end{lemma}

\begin{proof}
  The proof is by induction over $r$. For the base case $r=0$, suppose
  that Spoiler
  wins after round $0$. If $|V(G)|\neq|V(H)|$, the system $\GIeq{G,H}$
  has the following degree-$1$
  Nullstellensatz refutation:
  \[\textstyle
  \sum_{v\in V(G)}\frac{1}{a}\left(\sum_{w\in
      V(H)}x_{vw}-1\right)+
 \sum_{w\in V(H)}-\frac{1}{a}\left(\sum_{v\in
      V(G)}x_{vw}-1\right)=1,
  \]
  where $a=|V(G)-V(H)|$. It yields a degree-$1$ monomial PC
  refutation of $\GIeq{G,H}$ and thus a derivation of $\vec x_{\pi_0}$
  of degree $|\pi_0|\le k$.
  Otherwise, $\pi_0$ is not a local
  isomorphism. Then there is a 2-element subset
  $\pi:=\{(v,w),(v',w')\}\subseteq\pi_0$ that is
  not a local isomorphism. Multiplying the axiom $x_{vw}x_{v'w'}=\vec x_{\pi}$
  with the monomial $\vec x_{\pi_0\setminus\pi}$, we
  obtain a monomial-PC derivation of
  $\vec x_{\pi_0}$ of degree $|\pi_0|\le k$.

  For the inductive step, suppose that Spoiler has a winning strategy
  for the $(r+1)$-round game starting in position $\pi_0$. Let
  $\pi\subseteq\pi_0$ with $|\pi|<k$ be the set chosen by Spoiler in
  the first round of the game. We can derive $\vec x_{\pi_0}$ from
  $\vec x_{\pi}$ by multiplying with the monomial $\vec
  x_{\pi_0\setminus\pi}$. Hence it suffices to show that we can derive
  $\vec x_{\pi}$ in degree $k$.

  Consider
  the bipartite graph $B$ on $V(G)\uplus V(H)$ which has an edge $vw$ for all
  $v\in V(G),w\in V(H)$ such that Spoiler \emph{cannot} win from position
$\pi\cup\{(v,w)\}$ in at most $r$ rounds. As from position $\pi$, Spoiler wins in $r+1$ rounds, there is no bijection $f:V(G)\to V(H)$ such that
$(v,f(v))\in E(B)$ for all $v\in V(G)$. By Hall's Theorem, it follows that
there is a set $S\subseteq V(G)$ such that $|N^B(S)|<|S|$. Let
$S$ be a maximal set with this property and let $T:=N^B(S)$. 

We claim that $N^B(T)=S$. To see this, suppose for contradiction
that there is a vertex $v\in N^B(T)\setminus S$. By the maximality
of $S$, we have $N^B(v)\not\subseteq T$. Let $w\in N^B(v)\setminus
T$. Moreover, let $w'\in N^B(v)\cap T$ (exists because $v\in
N^B(T)$) and $v'\in N^B(w')\cap S$ (exists because
$T=N^B(S)$). Then by the definition of $B$, Duplicator has a winning
strategy for the $r$-round bijective $k$-pebble
game with initial positions $\pi\cup\{(v',w')\}$, $\pi\cup\{(v,w')\}$,
and $\pi\cup\{(v,w)\}$, which implies 
that she also has a winning strategy for the game with initial
position $\pi\cup\{(v',w)\}$. Here we use the fact that the relation
``duplicator has a winning strategy for the $r$-round bijective $k$-pebble
game'' defines an equivalence relation on the initial positions.
Thus $(v',w)\in E(B)$, which contradicts $w\not\in
N^B(S)$. This proves the claim.

By the induction hypothesis and the claim we know that ($\star$) $\vec
x_\pi x_{vw}$ has a degree-$k$ monomial PC derivation if $v\in
S,w\notin T$ or $v\notin S,w\in T$.  Furthermore, we can derive
\begin{align}
  \sum_{v\in S} \vec x_{\pi}\left(\sum_{w\in V(H)} x_{vw} - 1\right) -
  \sum_{w\in T} \vec x_{\pi}\left(\sum_{v\in V(G)} x_{vw} - 1\right)
\end{align}
by multiplying the axioms \eqref{eq:cont1}, \eqref{eq:cont2} with
$\vec x_\pi$ and building a linear combination.
By subtracting and adding monomials from ($\star$), this polynomial
simplifies to $(|T|-|S|)\vec x_\pi$. After dividing by the coefficient
$|T|-|S| \neq 0$, we get $\vec x_\pi$.
We can divide by $|T|-|S|$ because the characteristic of the field
$\mathbb F$ is $0$. 
\end{proof}

The following lemma is, at least implicitly, from
\cite{Grohe.2012}. As the formal framework is different there, we
nevertheless give a proof.

\begin{lemma}\label{lem:from_CkEquiv_to_solution}
  Let $\mathbb F$ be a field of characteristic $0$ and $k\ge 2$. If
  Duplicator has a winning strategy for the bijective $k$-pebble game
  on $G$, $H$ then there is a solution $\alpha$ of
  $\operatorname{MLIN}(\GIeq{G,H}^k)$ over $\mathbb F$ that additionally
  satisfies $\alpha(X_{\pi})=0 \Longrightarrow \alpha(X_{\rho})=0$ for all
  $\pi\subseteq\rho$.
\end{lemma}

\begin{proof}
  For all $\ell\le k$, we define an equivalence relation $\equiv^\ell$
  on $V(G)^\ell\cup V(H)^\ell$ as follows: for $I,I'\in\{G,H\}$ and
  $\ell$-tuples $\bar u=(u_1,\ldots,u_\ell)\in V(I)^\ell$, $\bar
  u'=(u'_1,\ldots,u'_\ell)\in V(I')^\ell$, we let $\bar
  u\equiv^\ell\bar u'$ if Duplicator has a winning strategy for the
  bijective $k$-pebble game on $(I,I')$
  with initial position
  $\{(u_1,u'_1),\ldots,(u_\ell,u'_\ell)\}$. We call the equivalence
  class of a tuple $\bar u$ the \emph{type} of $\bar u$ and denote it
  by $\tp(\bar u)$. Note that if $\tp(\bar u)=\tp(\bar u')$ then the
  mapping $\{(u_1,u'_1),\ldots,(u_\ell,u'_\ell)\}$ is a local
  isomorphism (of course the converse does not hold). 

  For $I\in\{G,H\}$ and $\bar u\in V(I)^\ell$ we let
  \[
  t(\bar u):=|\tp(\bar u)\cap V(I)^\ell|.
  \]
  It is easy to see that if $\bar u=(u_1,\ldots,u_\ell)\in V(I)^\ell$ and $\bar
  u'=(u'_1,\ldots,u'_{\ell'})\in V(I)^{\ell'}$, then
  \begin{equation}\label{eq:GO-0}
  \{u_1,\ldots,u_\ell\}=\{u'_1,\ldots,u'_{\ell'}\}\implies t(\bar
  u)=t(\bar u').
  \end{equation}
  In particular, the function $t$ is invariant under
  permutations. Also observe that for $\bar v\in V(G)^\ell$ and $\bar
  w\in V(H)^\ell$,
  \[
  \tp(\bar v)=\tp(\bar w)\implies t(\bar v)=t(\bar w).
  \]

  Now suppose that Duplicator has a winning strategy for the bijective
  $k$-pebble game on $G$, $H$. We define the desired solution $\alpha$
  to
  $\operatorname{MLIN}(\GIeq{G,H}^k)$ by
  \begin{align*}
    \alpha(X_\emptyset)&= 1\\
    \intertext{and for
      $\pi=\{(v_1,w_1),\ldots,(v_\ell,w_\ell)\}\subseteq V(G)\times
      V(H)^\ell$, where $\ell\le k$}
    \alpha(X_\pi)&=
    \begin{cases}\displaystyle
      \frac{1}{t(v_1,\ldots,v_\ell)}&\text{if
      }\tp(v_1,\ldots,v_\ell)=\tp(w_1,\ldots,w_\ell),\\
      0&\text{otherwise}.
    \end{cases}
  \end{align*}
  It follows from \eqref{eq:GO-0} that $\alpha$ is well-defined.
  We need to prove that it  satisfies the equations of $\operatorname{MLIN}(\GIeq{G,H}^k)$:
  \begin{align}
    \label{eq:GO-1}
    \sum_{v\in V(G)}X_{\pi\cup\{(v,w)\}}-X_{\pi}&=0&&\text{for all
      $w$ and $\pi$ of size $|\pi|\le k-1$},\\
    \label{eq:GO-2}
    \sum_{w\in V(H)}X_{\pi\cup\{(v,w)\}}-X_{\pi}&=0&&\text{for all
      $v$ and $\pi$ of size $|\pi|\le k-1$},\\
    \label{eq:GO-3}
    X_{\pi\cup\{(v,w),(v',w')\}}&=0&&\parbox[t]{5.5cm}{for all
      $v,v',w,w'$ such that $\{(v,w),(v',w')\}$ is no local
      isomorphism and all $\pi$ of size $|\pi|\le k-2$.}
  \end{align}
  The equations \eqref{eq:GO-3} are satisfied, because if
  $\{(v,w),(v',w')\}$ is no local isomorphism, then neither is
  $\pi\cup\{(v,w),(v',w')\}$. Assuming 
  $\pi=\{(v_1,w_1),\ldots,(v_{\ell},w_\ell)\}$ this implies
  $\tp(v_1,\ldots,v_\ell,v,v')\neq \tp(w_1,\ldots,w_\ell,w,w')$ and thus
  \[
  \alpha(X_{\pi\cup\{(v,w),(v',w')\}})=0.
  \]
  To see that the equations \eqref{eq:GO-1} are satisfied, let
  $\pi=\{(v_1,w_1),\ldots,(v_{\ell},w_\ell)\}$ for some $\ell\le k-1$
  and $w\in V(H)$. Let $\bar v=(v_1,\ldots,v_\ell)$ and $\bar
  w=(w_1,\ldots,w_\ell)$. If $\tp(\bar v)\neq\tp(\bar w)$ then
  \[
  \sum_{v\in
    V(G)}\alpha(X_{\pi\cup\{(v,w)\}})=\alpha(X_\pi)=0,
  \]
  and thus \eqref{eq:GO-1} is satisfied. Otherwise,
\begin{equation}\label{eq:GO-4}
  \alpha(X_\pi)=\frac{1}{t(\bar
    v)}=\frac{1}{t(\bar
    w)},
  \end{equation}
  and 
  \begin{align}
    \notag
  \sum_{v\in
    V(G)}\alpha(X_{\pi\cup\{(v,w)\}})&=\sum_{\substack{v\in V(G)\\\tp(\bar vv)=\tp(\bar
      ww)}}\alpha(X_{\pi\cup\{(v,w)\}})=\sum_{\substack{v\in
    V(G)\\\tp(\bar vv)=\tp(\bar ww)}}\frac{1}{t(\bar
ww)}\\
\label{eq:GO-5}
&=\frac{|\{v\in V(G)\mid \tp(\bar vv)=\tp(\bar ww)\}|}{t(\bar
ww)}.
\end{align}
We observe that for $\bar v'$ with $\tp(\bar v')=\tp(\vec w)$ we have 
\[
|\{v\in V(G)\mid \tp(\bar vv)=\tp(\bar ww)\}|=|\{v\in V(G)\mid
\tp(\bar v'v)=\tp(\bar ww)\}|.
\]
This follows from the properties of the bijective pebble game. Thus
\begin{equation}\label{eq:GO-6}
|\{v\in V(G)\mid \tp(\bar vv)=\tp(\bar ww)\}|=\frac{t(\bar ww)}{t(\bar w)}.
\end{equation}
Equations \eqref{eq:GO-4}--\eqref{eq:GO-6} imply that $\alpha$
satisfies \eqref{eq:GO-1}.

The proof that equations \eqref{eq:GO-2} are satisfied is symmetric.

Thus we have indeed defined a solution for the system
$\operatorname{MLIN}(\GIeq{G,H}^k)$. It remains to prove that this
solution satisfies $\alpha(X_{\pi})=0 \Longrightarrow \alpha(X_{\rho})=0$ for all
$\pi\subseteq\rho$. So let
$\rho=\{(v_1,w_1),\ldots,(v_m,w_m)\}$ and
$\pi=\{(v_1,w_1),\ldots,(v_\ell,w_\ell)\}\subseteq\rho$ for some
$\ell\le m\le k$. Then
\begin{align*}
  \alpha(X_{\pi})=0&\implies
  \tp(v_1,\ldots,v_\ell)\neq\tp(w_1,\ldots,w_\ell)\\
  &\implies \tp(v_1,\ldots,v_m)\neq\tp(w_1,\ldots,w_m)\\
  &\implies \alpha(X_{\rho})=0.
\end{align*}
\end{proof}

\begin{theorem}\label{thm:monomial}
Let $\mathbb F$ be a field of characteristic $0$. Then
  the following statements are equivalent for two graphs $G$ and $H$.
\begin{enumerate}
  \item The graphs are distinguishable by $k$-WL. 
  \item There is a degree-$k$ monomial-PC refutation of $\GIeq{G,H}$
    over $\mathbb F$.
\end{enumerate}
\end{theorem}
\begin{proof}
  Follows immediately from lemmas \ref{lem:monomialPC_MLIN},
  \ref{lem:from_Spoiler_to_monomial-PC} and
  \ref{lem:from_CkEquiv_to_solution}. 
\end{proof}

We do not now the exact relation between Nullstellensatz and monomial-PC for the graph isomorphism polynomials.
In particular, we do not know whether degree-$k$ Nullstellensatz is as
strong as the $k$-dimensional Weisfeiler-Lehman algorithm and leave
this as open question.
In the other direction, we remark that, at least for degree 2, full
polynomial calculus is strictly stronger than degree-2 monomial-PC and
hence the Colour Refinement Algorithm (see Example~\ref{exa:PC_is_stronger}).
However, we believe that the gap is not large. 
Our intuition is supported by Theorem~\ref{thm:main}, which implies that low-degree PC is not able to distinguish Cai-Fürer-Immerman graphs.
Thus, polynomial calculus has similar limitations as the Weisfeiler-Lehman algorithm \cite{caifurimm92}, 
Resolution \cite{Toran.2013}, 
the Sherali-Adams hierarchy \cite{atsman13,Grohe.2012} and the Positivstellensatz
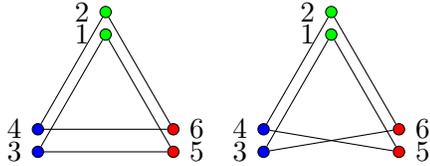
\cite{DWWZ.2013}.

\begin{tikzpicture}
  [scale=.6, knoten/.style={circle,draw=black,fill=black,
  inner  sep=0pt,minimum  size=1.5mm}]

  \node[knoten,label=left:{1},fill=green] (1) at (1.5,2.6) {};
  \node[knoten,label=left:{2},fill=green] (2) at (1.5,3.1) {};
  \node[knoten,label=left:{3},fill=blue] (3) at (0,0) {};
  \node[knoten,label=left:{4},fill=blue] (4) at (0,.5) {};
  \node[knoten,label=right:{5},fill=red] (5) at (3,0) {};
  \node[knoten,label=right:{6},fill=red] (6) at (3,.5) {};

 \draw[-] (4) -- (2);
 \draw[-] (3) -- (1);

 \draw[-] (2) -- (6);
 \draw[-] (1) -- (5);

 \draw[-] (4) -- (6);
 \draw[-] (3) -- (5);

\begin{scope}[xshift=5cm]

  \node[knoten,label=left:{1},fill=green] (1) at (1.5,2.6) {};
  \node[knoten,label=left:{2},fill=green] (2) at (1.5,3.1) {};
  \node[knoten,label=left:{3},fill=blue] (3) at (0,0) {};
  \node[knoten,label=left:{4},fill=blue] (4) at (0,.5) {};
  \node[knoten,label=right:{5},fill=red] (5) at (3,0) {};
  \node[knoten,label=right:{6},fill=red] (6) at (3,.5) {};

 \draw[-] (4) -- (2);
 \draw[-] (3) -- (1);

 \draw[-] (2) -- (6);
 \draw[-] (1) -- (5);

 \draw[-] (4) -- (5);
 \draw[-] (3) -- (6);

\end{scope}

\end{tikzpicture}

\begin{example}\label{exa:PC_is_stronger}
    Let $G$ be the disjoint union of two triangles and $H$ be a 6-cycle as depicted in the figure, where the vertices 1,2 are green, 3,4 are blue and 5,6 are red. 
    $G$ and $H$ cannot be distinguished by the 2-dimensional Weisfeiler-Lehman algorithm.
    Thus $\Piso(G,H)$ has no degree 2 monomial-PC refutation. 
    However, there is a degree 2 polynomial calculus refutation.
\end{example}

\begin{proof}
  Consider the axiom $\sum_{i=1}^6x_{1i}-1$. As $3,4,5,6$ have different colours than $1$ this simplifies to $x_{11}+x_{12}-1$.
  Multiplying $x_{11}+x_{12}-1$ with $x_{33}$ yields $x_{11}x_{33}+x_{12}x_{33}-x_{33}$ and hence ($\star$) $x_{11}x_{33}-x_{33}$ by subtracting the axiom $x_{12}x_{33}$.
  Similar, multiplying $x_{33}+x_{34}-1$ with $x_{11}$ yields $x_{33}x_{11}+x_{34}x_{11}-x_{11}$ and hence ($\star\star$) $x_{33}x_{11}-x_{11}$.
  Subtracting ($\star\star$) from ($\star$) yields (1) $x_{11}-x_{33}$.
  Note that we have obtained (1) by considering the edges between blue and green vertices. We can proceed the same way for the other two colour pairs to obtain (2) $x_{55}-x_{11}$ and (3) $x_{34}-x_{55}$.
  Now, adding (1), (2), and (3) yields (A) $x_{34}-x_{33}$. 
  We multiply with $x_{34}$ (this step is not allowed in monomial-PC) to get $x^2_{34}-x_{33}x_{34}$ which simplifies to ($\ast$) $x_{34}$ by adding the axiom $x_{33}x_{34}$ and subtracting $x^2_{34}-x_{34}$.
  In the same way we get ($\ast\ast$) $x_{34}$, by multiplying (A) with $x_{34}$.
  By subtracting ($\ast$) and ($\ast\ast$) from the simplified axiom $x_{33}+x_{34}-1$ and multiplying with $-1$ we have derived 1.
\end{proof}

\section{Groups CSPs and Tseitin Polynomials}

\subsection{From Group CSPs to Graph Isomorphism}
We start by defining a class of a constraint satisfaction problems
(CSPs) where the constraints are co-sets of certain
groups. %
Throughout this section, we let $\Gamma$ be a finite group.
Recall that a \emph{CSP-instance} has the form $(\CX,D,\CC)$, where $\CX$
is a finite set of \emph{variables}, $D$ is
a finite set called the \emph{domain} and $\CC$ a finite set of
\emph{constraints} of the form $(\bar x,R)$, where 
$\bar x\in \CX^k$ and
$R\subseteq D^k$, for some $k\ge 1$. A \emph{solution} to such an instance is an
assignment $\alpha:X\to D$ such that $\alpha(\bar x)\in R$ for all
constraints $(\bar x,R)\in\CC$.
An instance of a \emph{$\Gamma$-CSP} has domain $\Gamma$ and constraints of the form
$\big(\bar x,\Delta\gamma\big)$, where $\Delta\le\Gamma^k$ and
$\gamma\in\Gamma^k$. We specify instances as sets $\CC$ of constraints;
the variables are given implicitly. With each constraint
$C=\big((x_1,\ldots,x_k),\Delta\gamma\big)$, we associate the
\emph{homogeneous} constraint $\tilde
C=\big((x_1,\ldots,x_k),\Delta\big)$. For an instance
$\CC$, we let $\tilde{\CC}=\{\tilde C\mid C\in\CC\}$.

Next,  we reduce $\Gamma$-CSP to GI.
Let $\CC$ be a $\Gamma$-CSP in the variable set $\CX$. We construct a
coloured graph $G(\CC)$ as follows.
\begin{itemize}
\item For every variable $x\in\CX$ we take vertices $\gamma^{(x)}$ for all
  $\gamma\in \Gamma$. 
  We colour all these vertices with
  a fresh colour $L^{(x)}$.
\item For every constraint $C=((x_1,\ldots,x_k),\Delta\gamma)\in\CC$ we add
  vertices $\beta^{(C)}$ for all $\beta\in\Delta\gamma$. 
  We colour
  all these vertices with a fresh colour $L^{(C)}$. If
  $\beta=(\beta_1,\ldots,\beta_k)$, we add an edge
  $\{\beta^{(C)},\beta_i^{(x_i)}\}$ for all $i\in[k]$. We
    colour this edge with colour $M^{(i)}$.
\end{itemize}
We let $\tilde G(\CC)$ be the graph $G(\tilde\CC)$ where for all
constraints $C\in\CC$ we identify the
two colours $L^{(C)}$ and $L^{(\tilde C)}$.

\begin{lemma}\label{lem:isomorphicCSP}
  A $\Gamma$-CSP instance $\CC$ is satisfiable if and only if the graphs $G(\CC)$ and $\tilde G(\CC)$ are
  isomorphic.
\end{lemma}

\begin{proof}
  Let $G=(V,E):=G(\CC)$ and $\tilde G=(\tilde V,\tilde E):=\tilde G(\CC)$.
  Let $\phi:\CX\to\Gamma$ be a satisfying assignment for $\CC$. We define
  a mapping $f:V\to\tilde V$ as follows:
  \begin{itemize}
  \item For every $x\in\CX$ and $\gamma\in\Gamma$ we let
    $f(\gamma^{(x)}):=\big(\gamma\phi(x)^{-1}\big)^{(x)}$.
  \item For every $C=(x_1,\ldots,x_k,\Delta\gamma)\in\CC$ and every
    $\beta=(\beta_1,\ldots,\beta_k)\in\Delta\gamma$ we let 
    \[
    f(\beta^{(C)}):=\big(\beta_1\phi(x_1)^{-1},\ldots,\beta_k\phi(x_k)^{-1}\big)^{(C)}.
    \]
    To see that this is well defined, note that $\phi(\bar
    x):=\big(\phi(x_1),\ldots,\phi(x_k)\big)\in\Delta\gamma$, because
    $\phi$ satisfies the constraint $C$. Thus 
    \[\beta\phi(\bar
    x)^{-1}=\big(\beta_1\phi(x_1)^{-1},\ldots,\beta_k\phi(x_k)^{-1}\big)\in\Delta.
    \] 
  \end{itemize}
  It is easy to see that the mapping $f$ is bijective. To see that it
  is an isomorphism, consider, for some constraint
  $C=(x_1,\ldots,x_k,\Delta\gamma)\in\CC$ and some $i\in[k]$, a vertex
  $\beta^{(C)}$, where $\beta=(\beta_1,\ldots,\beta_k)\in\Delta\gamma$, and a vertex
  $\gamma^{(x_i)}$, where $\gamma\in\Gamma$. Then 
  \begin{align*}
  \{\beta^{(C)},\gamma^{(x_i)}\}\in
  E&\iff\beta_i=\gamma\iff\beta_i\phi(x_i)^{-1}=\gamma\phi(x_i)^{-1}\\
  &\iff
  \{f(\beta^{(C)}),f(\gamma^{(x_i)})\}\in \tilde E.
  \end{align*}
  To prove the backward direction, suppose that $f$ is an isomorphism
  from $G$ to $\tilde G$. We define an assignment $\phi:\CX\to\Gamma$ by
  \[
  \phi(x)^{(x)}=f^{-1}(1^{(x)}).
  \]
  (Here $1^{(x)}$ denotes the $x$-copy of the
  unit element $1\in\Gamma$ in the graph $\tilde G$.) To see that
  $\phi$ is a satisfying assignment, consider a constraint
  $C=(x_1,\ldots,x_k,\Delta\gamma)\in\CC$. Let
  $\beta=(\beta_1,\ldots,\beta_k)$ with $\beta_i=\phi(x_i)$. 

  We need to prove that $\beta\in\Delta\gamma$. We have
  $f(\beta_i^{(x_i)})=1^{(x_i)}$. As $\bar 1=(1,\ldots,1)\in\Delta$,
  the vertex $\bar 1^{(\tilde C)}\in \tilde V$ has edges to all
  vertices $f(\beta_i^{(x_i))})$. Thus the vertex $f^{-1}(\bar
  1^{(\tilde C)})$ has colour $L^{(C)}=L^{(\tilde C)}$ and edges to
  the vertices $\beta_i^{(x_i)}$. This implies that $f^{-1}(\bar
  1^{(\tilde C)})=\alpha^{(C)}$ for some $\alpha\in\Delta\gamma$ and
  $\alpha=(\beta_1,\ldots,\beta_k)=\beta$.
\end{proof}

\begin{remark}
  The lemma shows that our construction of graphs $G(\CC)$ and $\tilde
  G(\CC)$ from a $\Gamma$-CSP gives a reduction from $\Gamma$-CSPs
  to coloured graph isomorophism. Clearly, this is a polynomial time
  reduction. Observe that for a fixed group $\Gamma$ and a fixed-arity
  $k$, the graphs $G(\CC)$ and $\tilde
  G(\CC)$ have a \emph{bounded colour class size}, that is, there is a
  bound (of $|\Gamma|^k$) on the maximum number of vertices of each
  colour. It is long known that the isomorphism problem for graphs of
  bounded colour class size is solvable in polynomial time
  \cite{bab79,furhopluk80}.

  Interestingly, there is also a converse reduction. Let $G,H$ be a
  pair coloured graph where the colour classes have size at most
  $\ell$. Suppose that there are $m$ colours, and let $C_i(G)$ and
  $C_i(H)$ be the vertices of $G,H$, of colour $i$. Without loss of
  generality we may assume that $|C_i(G)|=|C_i(H)|=:\ell_i$, where
  $\ell_i\le\ell$. Suppose that
  $C_i(G)=\{v_{i1},\ldots,v_{i\ell_i}\}$ and
  $C_i(H)=\{w_{i1},\ldots,w_{i\ell_i}\}$. Then an isomorphisms $g$ from $G$
  to $H$ can be described as tuples $(\gamma_1,\ldots,\gamma_m)$ of
  mappings $\gamma_i:[\ell]\to[\ell]$: we let
  $g(v_{ij})=w_{i\gamma_i(j)}$; conversely for a given isomorphism $g$
  we choose $\gamma_i(j)$ to be the $j'$ such that $g(v_{ij})=w_{ij'}$
  if $1\le j\le\ell_i$, and we let $\gamma_i(j):=j$ if $\ell_i+1\le
  j\le\ell$.

  This enables us to describe isomorphisms as solutions to an $S_\ell$-CSP
  instance (where
  $S_\ell$ denotes the symmetric group on $[\ell]$), with variables
  $x_1,\ldots,x_m$ and the following constraints:
  \begin{itemize}
  \item for all $i\in[m]$ a unary constraint $(x_i,P_i)$, where
    \[
    P_i:=\big\{\gamma_i\in S_\ell\bigmid\gamma_i(j)=j\text{ for all }j>\ell_i\};
    \]
  \item for all $i,i'\in[m]$  (not necessarily distinct)  a binary constraints $((x_i,x_{i'}),R_{ii'})$ with
  \begin{align*}
  R_{ii'}&=\Big\{(\gamma_i,\gamma_{i'})\in S_\ell^2\Bigmid \forall
  j\in[\ell_i],j'\in[\ell_{i'}]:\\
  &\hspace{2cm}\big(v_{ij}v_{i'j'}\in E(G)\iff
  w_{i\gamma_i(j)}w_{i'\gamma_{i'}(j')}\in E(H)\Big\}.
  \end{align*}
  \end{itemize}
  Observe that the $P_i$ are subgroups of $S_\ell$ and the $R_{ii'}$
  are cosets of some subgroup of $S_\ell^2$ (essentially the
  automorphism group of the subgraph of $G$ with vertices in
  $C_i(G)\cup C_{i'}(G)$ and all edges between the classes). Hence
  this really defines an $S_\ell$-CSP instance. It is easy to see that
  solutions to this CSP instance correspond to isomorphims between $G$
  and $H$.

  Dawar \cite{daw14} observed that $\Gamma$-CSPs have a constraint language that
  admits Mal'tsev polymorphisms (see~\cite{buldal06}). Such CSPs are
  known to be solvable in polynomial time~\cite{bul02,buldal06}. So we
  obtain a reduction from bounded colour class graph isomorphism to
  constraint satisfaction for constraint languages with Mal'tsev
  polymorphisms. Such a reduction (essentially the same one as ours) has
  also been given in \cite{klilasoch+14}.
\end{remark}

\begin{example}[The Tseitin Tautologies and the CFI-construction]\label{exa:tseitin}
  For every graph $H$ and set $T\subseteq V(H)$ we define the following $\mathbb Z_2$-CSP $\TS=\TS(H,T)$.
  \begin{itemize}
  \item For every edge $e\in E(H)$ we have a variable $z_e$.
  \item For every vertex $v\in V(H)$ we define a constraint
    $C_v$. Suppose that $v$ is incident with the edges
    $e_1,\ldots,e_k$ (in an arbitrary order), and let
    $z_i:=z_{e_i}$. Let 
    $
    \Delta:=\{(i_1,\ldots,i_k)\in\mathbb
    Z_2^n\mid\sum_{i=1}^ki_j=0\}\le\mathbb Z_2^k.
    $
    We will also use the coset
    $
    \Delta+(1,0,\ldots,0)=\{(i_1,\ldots,i_k)\in\mathbb
    Z_2^n\mid\sum_{i=1}^ki_j=1\}
    $
    If $v\not\in T$, we
    let 
    $
    C_v:=\big(z_1,\ldots,z_k,\Delta),
    $
    and if $v\in T$ we let $C_v:=\big(z_1,\ldots,z_k,\Delta+(1,0,\ldots,0))$.
  \end{itemize}
  Observe that $\TS$ is a set of Boolean constraints, all
  of them
  linear equations over the field $\mathbb F_2$; they are known as the
  \emph{Tseitin tautologies} associated with $H$ and $T$. We think of
  assigning a ``charge'' $1$ to every vertex in $T$ and charge $0$ to
  all remaining vertices. Now we are looking for a set $F\subseteq
  E(H)$ of edges such that for every vertex $v$, the number of edges
  in $F$ incident with $v$ is congruent to the charge of $v$ modulo
  $2$. A simple double counting argument shows that $\TS$ is
  unsatisfiable if $|T|$ is odd. (The sum of degrees in
  the graph $(V(H),F)$ is even and, by construction, of the same parity as the sum $|T|$
  of the charges, which is odd.)

    It turns out that the graphs $G(\TS)$ and $\tilde G(\TS)$ are
  precisely the \emph{CFI-graphs} defined from $H$ with all
  vertices in $T$ ``twisted''. These graphs have been introduced by
  Cai, Fürer, and Immerman~\cite{caifurimm92} to prove lower bounds
  for the Weisfeiler-Lehman algorithm and have found various other
  applications in finite model theory since then.
\end{example}

\subsection{Low-Degree Reduction From Tseitin to Isomorphism}
For every graph $H$ and set $T\subseteq V(H)$, we let
${\PTs}(H,T)$ be the following system of
polynomial equations:
\begin{align}
  \label{eq:ts1}
  z_e^2-1&=0&&\text{for all }e\in E(H),\\
  \label{eq:ts2}
  1+z_{e_1}z_{e_2}\cdots z_{e_k}&=0&&\text{for all }v\in T\text{  with incident edges }e_1,\ldots,e_k,\\
  \label{eq:ts3}
  1-z_{e_1}z_{e_2}\cdots z_{e_k}&=0&&\text{for all }v\in V(H)\setminus
  T\text{  with incident edges }e_1,\ldots,e_k.
\end{align}
Observe that for every field $\mathbb F$ of characteristic $\neq 2$
there is a one-to-one correspondence between solutions to the system
${\PTs}(H,T)$ over $\mathbb F$ and solutions for the
CSP-instance $\TS(H,T)$ (see Example~\ref{exa:tseitin}) via the
``Fourier'' correspondence $1\mapsto0,-1\mapsto 1$.

\begin{lemma}\label{lem:tsei2iso}
  Let $\mathbb F$ be a field of characteristic $\neq 2$. Let $k\ge 2$
  be even and $H$ a $k$-regular graph, and let $T\subseteq V(H)$. Let
  $G:=G(\TS(H,T))$ and $\tilde G:=\tilde G(\TS(H,T))$. 

  Then there is a degree-$(k,2k)$ reduction from ${\PTs}(H,T)$
  to $\Piso(G,H)$.
\end{lemma}

\begin{proof}
Let us first simplify the notation. We let $\TS:=\TS(H,T)$ and ${\PTs}:={\PTs}(H,P)$
  and ${\Piso}:={\Piso}(H,P)$

We
denote the vertices of $H$ by $t,u$, the vertices of $G$ by $v,w$ and
the vertices of $\tilde G$ by $\tilde v,\tilde w$. It will be
convenient to view the CSP $\TS$ as a $\Gamma$-CSP for the
multiplicative group $\Gamma=(\{1,-1\},\cdot)$. Then the constraint
$C_t$ associated with vertex $t$ is
$C_t=((z_{e_1},\ldots,z_{e_k}),{Z}_t)$, where 
\[
{Z}_t=
\begin{cases}
  \big\{(\zeta_1,\ldots,\zeta_k)\in\{1,-1\}^k\bigmid\prod_{i=1}^k\zeta_i=-1\big\}&\text{if
  }t\in T,\\
  \big\{(\zeta_1,\ldots,\zeta_k)\in\{1,-1\}^k\bigmid\prod_{i=1}^k\zeta_i=1\big\}&\text{if
  }t\in V(H)\setminus T.  
\end{cases}
\]
We let 
\[
\tilde
Z_t=\big\{(\zeta_1,\ldots,\zeta_k)\in\{1,-1\}^k\bigmid\prod_{i=1}^k\zeta_i=1\big\}
\]
for all $t\in V(H)$. Note that $\tilde Z_t={Z}_t$ precisely for the $t\in V(H)\setminus T$.

The graphs $G$ and $\tilde G$ have vertices $1^{(z_e)}$ and
$-1^{(z_e)}$ for $e\in E(H)$, which in the following we denote by
$1^{(e)}$ and $-1^{(e)}$. Furthermore, the
graphs have vertices $\bzeta^{(C_t)}$ for all $t\in V(H)$ and
$\bzeta\in{Z}_t$ or $\bzeta\in\tilde Z_t$, which in the
following we denote by $\bzeta^{(t)}$. If $\bzeta=(\zeta_1,\ldots,\zeta_k)$ and
$t$ is incident with $e_1,\ldots,e_k$, for
all $i\in[k]$ we have an $M^{(i)}$-coloured edge
$\bzeta^{(t)}\zeta_i^{(e_i)}$. 

To avoid excessive indexing, we write $z(e)$ instead of $z_e$ and
$x(v,\tilde v)$ instead of $y
x_{v\tilde v}$ to denote the variables of
the polynomials in ${\PTs}(H,T)$ and ${\Piso}(G,H)$. 

Now we are ready to define the reduction.
  Let us first define the polynomials $f_x$ for the variables $x$ of
  ${\Piso}$.
  \begin{itemize}
  \item If $x=x(\zeta^{(e)},\eta^{(e)})$, for some $e\in E(H)$ and $\zeta,\eta\in\{1,-1\}$, we let
    \[
    f_{x}(z(e))=\frac{1}{2}\big(1-\zeta\eta z(e)\big).
    \]
  \item If $x=x(\bzeta^{(t)},\beta^{(t)})$ for some $t\in V(H)$, incident with edges $e_1,\ldots,e_k\in
    E(H)$, and 
    $\bzeta=(\zeta_1,\ldots,\zeta_k)\in{Z}_t,\bseta=(\eta_1,\ldots,\eta_k)\in\tilde{Z}_t$
    we let
    \[
    f_{x}(z(e_1),\ldots,z(e_k))=\prod_{i=1}^k\frac{1}{2}\big(1-\zeta_i\eta_i z(e_i)\big).
    \]
  \item For all other variables $x$, we let $f_x=0$.
  \end{itemize}
  Now we need to prove that   for each polynomial
  $q(x_1,\ldots,x_\ell)\in{\Piso}\cup\SQ$ the polynomial $
  q\big(f_{x_1}(\bar z),\ldots,f_{x_\ell}(\bar z)\big)$ has a degree
  $2k$ derivation from ${\PTs}$.

  \paragraph*{The polynomials $x(v,\tilde v)^2-x(v,\tilde v)$ for $v\in V(G),\tilde v\in
    V(\tilde G)$}~\\
  We only have to consider the cases 
  \begin{eroman}
    \item $v=\zeta^{(e)}$, $\tilde v=\eta^{(e)}$ for some $e \in E(H)$ and
  $\zeta,\eta\in\{-1,1\}$, 
  \item $v=\boldsymbol\zeta^{(t)}$, $\tilde v=\eta^{(t)})$ for some
    $t\in V(H)$ and
  $\bzeta\in Z_t$, $\bseta\in\tilde Z_t$.
  \end{eroman}
  For all other pairs $v\in V(G),\tilde v\in V(\tilde G)$ we have
  $f_{x(v,\tilde v)}=0$ and thus $f_{x(v,\tilde v)}^2-f_{x(v,\tilde v)}=0$,
  which makes it trivially derivable in the polynomial calculus.

  Let us consider case (i) first. We have
  \begin{align}
    \notag
  f_{x(v,\tilde v)}(z(e))^2-f_{x(v,\tilde v)}(z(e))
  &=\frac{1}{4}-\frac{1}{2}\zeta\eta
  z(e)+\frac{1}{4}\zeta^2\eta^2z(e)^2-\frac{1}{2}+\frac{1}{2}\zeta\eta
  z(e)\\
  \label{eq:tsiso1}
  &=\frac{1}{4}\big(z(e)^2-1\big),
  \end{align}
  where the last equality holds because $\zeta^2=\eta^2=1$. As
  $z(e)^2-1$ is in ${\PTs}$, this gives us a trivial degree-$2$
  derivation of $f_{x(v,\tilde v)}(z(e))^2-f_{x(v,\tilde v)}(z(e))$.

  Let us now consider case (ii). Suppose that $t$ is incident with the
  edges $e_1,\ldots,e_k$ and that
  $\bzeta=(\zeta_1,\ldots,\zeta_k)$ and
  $\bseta=(\eta_1,\ldots,\eta_k)$. For $i\in[k]$, we let $z_i=z(e_i)$
  and $f_i(z_i)=\frac{1}{2}(1-\zeta_i\eta_i z_i)$. As in
  \eqref{eq:tsiso1},
  \begin{equation}
    \label{eq:tsiso2}
    f_i(z_i)^2-f_i(z_i)=\frac{1}{4}(z_i^2-1).
  \end{equation}
  We have
  \[
    f_{x(v,\tilde v)}(z_1,\ldots,z_k)^2- f_{x(v,\tilde
      v)}(z_1,\ldots,z_k)^2
    =\prod_{i=1}^kf_i(z_i)^2-\prod_{i=1}^kp(z_i)
  \]
  We prove that we can derive
  $\prod_{i=1}^jf_i(z_i)^2-\prod_{i=1}^jp(z_i)$ by induction on
  $j$. For $j=1$, this follows from \eqref{eq:tsiso2}. For the
  inductive step $j-1\to j$, we write
  \begin{align*}
  &\prod_{i=1}^jf_i(z_i)^2-\prod_{i=1}^jf_i(z_i)\\
  =\,&f_j(z_j)^2\prod_{i=1}^{j-1}f_i(z_i)^2-f_j(z_j)\prod_{i=1}^{j-1}f_i(z_i)^2+f_j(z_j)\prod_{i=1}^{j-1}f_i(z_i)^2-f_j(z_j)
  \prod_{i=1}^{j-1}f_i(z_i)\\
  =\,&(f_j(z_j)^2-f_j(z_j))\prod_{i=1}^{j-1}f_i(z_i)^2+f_j(z_j)\Big(\prod_{i=1}^{j-1}f_i(z_i)^2-\prod_{i=1}^{j-1}f_i(z_i)\Big).
  \end{align*}
  If follows from \eqref{eq:tsiso2} that $(f_j(z_j)^2-f_j(z_j))$ can
  be derived. Thus $(f_j(z_j)^2-f_j(z_j))\prod_{i=1}^{j-1}f_i(z_i)^2$
  can be derived as well. It follows from the induction hypothesis
  that $\prod_{i=1}^{j-1}f_i(z_i)^2-\prod_{i=1}^{j-1}f_i(z_i)$ can be
  derived. Thus
  $f_j(z_j)\Big(\prod_{i=1}^{j-1}f_i(z_i)^2-\prod_{i=1}^{j-1}f_i(z_i)\Big)$
  can be derived as well, which implies that
  $\prod_{i=1}^jf_i(z_i)^2-\prod_{i=1}^jf_i(z_i)$ can be derived. As
  none of the polynomials involved in these derivations has degree
  greater than $2k$, this gives us a degree-$2k$ derivation.

  \paragraph*{The polynomials $\sum_{\tilde v\in V(\tilde G)}x(v,\tilde
    v)-1$ 
    for $v\in V(G)$}~\\
  Suppose first that $v=\zeta^{(e)}$ for some $e\in E(H)$ and
  $\zeta\in\{1,-1\}$, and let $\eta=-\zeta$. Then
    \[
    \sum_{\tilde v\in V(\tilde G)}\hspace{-0.8em}f_{x(v,\tilde v)}(\bar z)-1
    =f_{\zeta^{(e)},\zeta^{(e)}}(z(e))+f_{\zeta^{(e)},\eta^{(e)}}(z(e))-1
    =-\frac{1}{2}\zeta^2z(e)-\frac{1}{2}\zeta\eta z(e)=0,
    \]
    which is trivially derivable.

    Suppose next that $v=\bzeta^{(t)}$ for some $t\in V(H)$ and
    $\bzeta=(\zeta_1,\ldots,\zeta_k)\in Z_t$. Suppose
    that $t$ is incident with edges $e_1,\ldots,e_k\in E(H)$, and let
    $z_i=z(e_i)$. Then
    \begin{align}
      \notag
    \sum_{\tilde v\in V(\tilde G)}f_{x(v,\tilde v)}(\bar z)-1
    &=\sum_{\bseta\in\tilde Z_t}
    f_{x(\bzeta^{(t)},\bseta^{(t)})}(z_1,\ldots,z_k)-1\\
    &\label{eq:tsiso3}
    =-1+\frac{1}{2^{k}}\sum_{(\eta_1,\ldots,\eta_k)\in\tilde Z_t}
    \prod_{i=1}^k(1-\zeta_i\eta_iz_i)
    \end{align}
    Observe that 
    \begin{equation}\label{eq:tsiso4}
    \sum_{(\eta_1,\ldots,\eta_k)\in\tilde Z_t}
    \prod_{i=1}^k(1-\zeta_i\eta_iz_i)
    =\sum_{(\eta_1,\ldots,\eta_k)\in Z_t}
    \prod_{i=1}^k(1-\eta_iz_i)
    \end{equation}

    \begin{claim}
      Let $\ell\ge 1$, $\epsilon\in\{1,-1\}$, and
      $Z({\ell},\epsilon)=\big\{(\eta_1,\ldots,\eta_{\ell})\bigmid\prod_{i=1}^{\ell}\eta_i=\epsilon\big\}$. Then
      \[
      \sum_{(\eta_1,\ldots,\eta_{\ell})\in Z({\ell},\epsilon)}
      \prod_{i=1}^{\ell}(1-\eta_iz_i)=2^{{\ell}-1}\Big(1+(-1)^{\ell}\epsilon\prod_{i=1}^{\ell}z_i\Big).
      \]
      \proof
      We prove the claim by induction on ${\ell}$. For ${\ell}=1$ it is
      trivial. For the inductive step ${\ell}-1\to {\ell}$, we write
      \begin{align*}
      &\sum_{(\eta_1,\ldots,\eta_{\ell})\in Z({\ell},\epsilon)}
      \prod_{i=1}^{\ell}(1-\eta_iz_i)\\
      =\,&(1-z_{\ell})\hspace{-2em}\sum_{(\eta_1,\ldots,\eta_{{\ell}-1})\in Z({\ell}-1,\epsilon)}
      \prod_{i=1}^{{\ell}-1}(1-\eta_iz_i)
      +(1+z_{\ell}) \hspace{-2em}\sum_{(\eta_1,\ldots,\eta_{{\ell}-1})\in Z({\ell}-1,-\epsilon)}
      \prod_{i=1}^{{\ell}-1}(1-\eta_iz_i)\\
      =\,&2^{{\ell}-2}\Big((1-z_{\ell})\big(1+(-1)^{{\ell}-1}\epsilon\prod_{i=1}^{{\ell}-1}z_i\big)+(1+z_{\ell})\big(1-(-1)^{{\ell}-1}\epsilon\prod_{i=1}^{{\ell}-1}z_i\big)\Big)\\
      =\,&2^{{\ell}-2}\Big(2+2(-1)^{\ell}\epsilon\prod_{i=1}^{\ell}z_i\Big)\\
      =\,&2^{{\ell}-1}\Big(1+(-1)^{\ell}\epsilon\prod_{i=1}^{\ell}z_i\Big)
      \uenda
      \end{align*}
    \end{claim}

    Let $\epsilon_t=-1$ if $t\in T$ and $\epsilon_t=1$ if $t\not\in
    T$. As $k$ is even, by the claim we have 
    \[
    \sum_{(\eta_1,\ldots,\eta_k)\in Z_t}
    \prod_{i=1}^k(1-\eta_iz_i)=2^{k-1}\Big(1+\epsilon_t\prod_{i=1}^kz_i\Big).
    \]
    Then from \eqref{eq:tsiso3} and
    \eqref{eq:tsiso4} we obtain
        \begin{equation}\label{eq:tsiso5}
    \sum_{\tilde v\in V(\tilde G)}f_{x(v,\tilde v)}(\bar
    z)-1=-1+\frac{1}{2}\Big(1+\epsilon_t\prod_{i=1}^kz_i\Big)=-\frac{1}{2}\Big(1-\epsilon_t\prod_{i=1}^kz_i\Big).
    \end{equation}
    As $1-\epsilon_t\prod_{i=1}^kz_i\in{\PTs}$, this polynomial is derivable.

  \paragraph*{The polynomials $\sum_{v\in V(G)}x(v,\tilde
    v)-1$ 
    for $\tilde v\in V(\tilde G)$}~\\
  This case is symmetric to the previous one. 

  \paragraph*{The polynomials $x(v_1,\tilde v_1)x(v_2,\tilde v_2)$ for
    $v_1,v_2\in V(G), \tilde v_1,\tilde v_2\in V(\tilde G)$ such that
    $\{(v_1,\tilde v_1),(v_2,\tilde v_2)\}$ is not a local isomorphism}~\\
  As $f_{x(v_,\tilde v_i)}=0$ unless $v_i=\zeta^{(e)}$ and
  $\tilde v_i=\eta^{(e)}$ $e\in E(H)$ and $\zeta,\eta\in\{-1,1\}$, or
  $v_i=\bzeta^{(t)}$ and $\tilde v_i=\bseta^{(t)}$ for some
  $t\in V(H)$ and $\bzeta\in Z_t,\bseta\in\tilde Z_t$, we assume that
  this is the case for $i=1,2$. In the former
  case we them \emph{$e$-vertices} and in the latter
  \emph{$t$-vertices}. In order for the mapping
  $v_1v_2\mapsto\tilde v_1\tilde v_2$ to be no local isomorphism,
  the following may happen:
  \begin{eroman}
    \item $v_1=v_2$ and $\tilde v_1\neq \tilde v_2$;
    \item $v_1\neq v_2$ and $\tilde v_1=\tilde v_2$;
    \item $v_1v_2\in E(G)$ and $\tilde v_1\tilde v_2\not\in E(H)$;
    \item $v_1v_2\not\in E(G)$ and $\tilde v_1\tilde v_2\in E(H)$.
  \end{eroman}
  By symmetry, it suffices to consider cases (i) and (iii).

  In case (i), assume first that $v_1,\tilde v_1$ are
  $e$-vertices. Then $v_2=v_1,\tilde v_2$ are $e$-vertices as
  well. Say, $v_1=v_2=\zeta^{e}$ and $\tilde v_1=\eta_1^{(e)}, \tilde
  v_2=\eta_2^{(e)}$ for $\eta_1\neq\eta_2$. Then
  \begin{align*}
  f_{x(v_1,\tilde v_1)}f_{x(v_2,\tilde v_2)}&=\frac{1}{4}(1-\zeta^2
  z(e))(1-\eta_1\eta_2
  z(e))\\
    &=\frac{1}{4}(1-z(e))(1+z(e))=\frac{1}{4}(1-z(e)^2).
  \end{align*}
  As $z(e)^2-1\in{\PTs}$, this polynomial is derivable.
  
  Assume next that $v_1,\tilde v_1$ are
  $t$-vertices. Then $v_2=v_1,\tilde v_2$ are $t$-vertices as
  well. Say, $v_1=v_2=\bzeta^{(t)}$ for some
  $\bzeta=(\zeta_1,\ldots,\zeta_k)\in Z_t\cap\tilde Z_t$ and $\tilde v_i=\bzeta_i^{(t)}$ for some
  $\bseta_i=(\eta_{i1},\ldots,\eta_{ik})\in \tilde Z_t$. Then
  $\bseta_1\neq\bseta_2$. Say, $\eta_{1k}\neq\eta_{2k}$. For all $j\in[k]$, let
  $z_j=z(e_j)$. We have
  \begin{align*}
  f_{x(v_1,\tilde v_1)}(\bar z)\cdot f_{x(v_2,\tilde v_2)}(\bar z)
  &=\Big(\prod_{j=1}^k\frac{1}{2}(1-\zeta_i^2 z_j)\Big)\cdot \Big(\prod_{j=1}^k\frac{1}{2}(1-\eta_{1j}\eta_{2j} z_j)\Big)\\
  &=\frac{1}{2^{2k}}\Big(\prod_{j=1}^{k-1}(1-\zeta_j^2z_j)(1-\eta_{1j}\eta_{2j}z_j)\Big)
    (1-z_k)(1+z_k)\\
  &=\frac{1}{2^{2k}}\Big(\prod_{j=1}^{k-1}(1-\zeta_j^2z_j)(1-\eta_{1j}\eta_{2j}z_j)\Big)
    (1-z_k^2).
  \end{align*}
  As $z_k^2-1\in{\PTs}$, this polynomial is derivable.

  In case (iii), 
  $v_i,\tilde v_i$ must be $t$-vertices and $v_{3-i},\tilde v_{3-i}$
  must be $e$-vertices for some $i\in[2]$ and $t\in V(H), e\in E(H)$
  such that $e$ is incident with $t$, because otherwise there will be
  no edges between either $v_1$ and $v_2$ or $\tilde v_1$ and
  $\tilde v_2$.

  Say, $v_1=\bzeta^{(t)}, \tilde v_1=\bseta^{(t)}$ for some $t\in
  V(H)$ and $\bzeta=(\zeta_1,\ldots,\zeta_k)\in Z_t$,
  $\bseta=(\eta_1,\ldots,\eta_k)\in \tilde Z_t$ and
  $v_2=\zeta^{(e)},\tilde v_2=\eta^{(e)}$ for some $e\in E(H)$
  incident with $v$ and $\zeta,\eta\in\{-1,1\}$. Let $e_1,\ldots,e_k$
  be the edges incident with $t$, and assume that $e=e_k$. 
As $v_1v_2\in E(G)$, we have $\zeta=\zeta_k$, and as $\tilde v_1\tilde v_2\not\in
  E(\tilde G)$,
  we have $\eta\neq\eta_k$. This implies
  $\zeta\eta\neq\zeta_k\eta_k$. For all $j\in[k]$, let
  $z_j=z(e_j)$. We have
  \begin{align*}
  f_{x(v_1,\tilde v_1)}(\bar z)\cdot f_{x(v_2,\tilde v_2)}(\bar z)
  &=\Big(\prod_{i=1}^k\frac{1}{2}(1-\zeta_i\eta_i
  z_i)\Big)\cdot\frac{1}{2}(1-\zeta\eta z_k)\\
  &=\frac{1}{2^{k+1}}\prod_{i=1}^{k-1}(1-\zeta_i\eta_i
  z_i)(1-z_k)(1+z_k)\\
  &=-\frac{1}{2^{k+1}}\prod_{i=1}^{k-1}(1-\zeta_i\eta_i
  z_i)(z_k^2-1).
  \end{align*}
  As $z_k^2-1\in{\PTs}$, this polynomial is derivable.
\end{proof}

\section{Lower Bounds}

We obtain our lower bounds combining the low-degree reduction of the
previous section with known lower bounds for Tseitin
polynomials due to Buss et al. \cite{Buss.2001} for polynomial calculus and Grigoriev \cite{Grigoriev.2001a} for the Positivstellensatz calculus.

\begin{theorem}[\cite{Buss.2001,Grigoriev.2001a}]\label{thm:Tseitin}
  For every $n\in\mathbb N$ there is a 
  6-regular graph $H_n$ of size $O(n)$ such that $\PTs(H_n,V(H_n))$ is
  unsatisfiable, but:
\begin{enumerate}
  \item there is no degree-$n$ polynomial calculus refutation of
    $\PTs(H_n,V(H_n))$ over any field $\FF$ of characteristic $\neq 2$;
  \item there is no degree-$n$ Positivstellensatz calculus refutation of $\PTs(H_n,V(H_n))$ over the reals.
\end{enumerate}
\end{theorem}

Now our main lower bound theorem reads as follows.

\begin{theorem}\label{thm:main}
  For every $n\in\mathbb N$ there are non-isomorphic graphs $G_n$, $\tilde G_n$ of
  size $O(n)$, such that
\begin{enumerate}
  \item there is no degree-$n$ polynomial calculus refutation of $\Piso(G_n,\tilde G_n)$ over any field $\FF$ of characteristic $\neq 2$; 
  \item there is no degree-$n$ Positivstellensatz calculus refutation of $\Piso(G_n,\tilde G_n)$ over the reals.
\end{enumerate}
\end{theorem}

\begin{proof}
  This follows from Lemmas~\ref{lem:lowdeg} and \ref{lem:tsei2iso} and Theorem~\ref{thm:Tseitin}.
\end{proof}

It follows that over finite fields, polynomial calculus has similar shortcomings than over fields of characteristic 0. 
However, a remarkable exception is $\FF_2$, where we are not able to prove linear lower bounds on the degree.
Here the approach to reduce from Tseitin fails, as the Tseitin Tautologies are satisfiable over $\FF_2$.
As a matter of fact, the next theorem shows that CFI-graphs can be
distinguished with Nullstellensatz of degree $2$ over $\FF_2$.

\begin{theorem}\label{thm:F2_solves_CFI}
  Let $H$ be a graph $T\subseteq V(H)$ such that $|T|$ is odd.  Then
  there is a degree-2 Nullstellensatz refutation over $\FF_2$ of
  $\Piso(G,\tilde G)$, where $G=G(\TS(H,T))$ and $\tilde G=\tilde
  G(\TS(H,T))$.
\end{theorem}

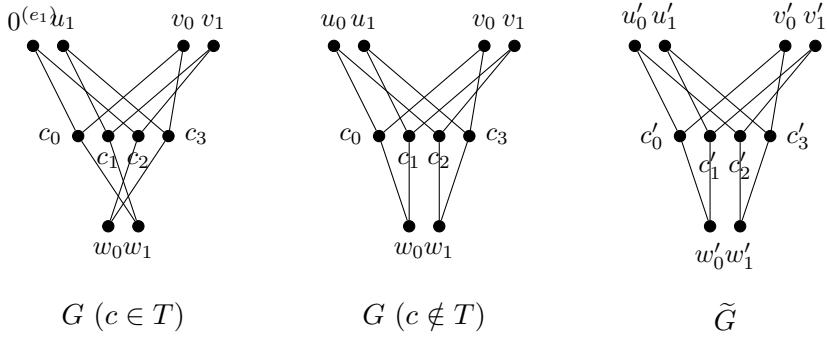
\begin{figure}
  
\begin{tikzpicture}
[knoten/.style={circle,draw=black,fill=black,
  inner  sep=0pt,minimum  size=1.5mm},scale=2]

 \node at (11,-0.2) {$G$ ($c\in T$)};
\node[knoten] (andO1) at (10.5-0.1,1.6) [label=above:{\small$0^{(e_1)}$}] {};
\node[knoten] (andI1) at (10.5+0.1,1.6) [label=above:{\small$u_1$}]{};
\node[knoten] (andO2) at (11.5-0.1,1.6) [label=above:{\small$v_0$}]{};
\node[knoten] (andI2) at (11.5+0.1,1.6) [label=above:{\small$v_1$}]{};
\node[knoten] (andI3) at (11-0.1,0.4) [label=below:{\small$w_0$}]{};
\node[knoten] (andO3) at (11+0.1,0.4) [label=below:{\small$w_1$}]{};
\node[knoten] (CFI1) at (11-0.3,1) [label=left:{\small$c_0$}] {};
\node[knoten] (CFI2) at (11-0.1,1) [label=below:{\small$c_1$}]{};
\node[knoten] (CFI3) at (11+0.1,1) [label=below:{\small$c_2$}]{};
\node[knoten] (CFI4) at (11+0.3,1) [label=right:{\small$c_3$}]{};
\draw[-] (andO3) -- (CFI1);
\draw[-] (andO3) -- (CFI2);
\draw[-] (andI3) -- (CFI3);
\draw[-] (andI3) -- (CFI4);
\draw[-] (andO1) -- (CFI1);
\draw[-] (andO1) -- (CFI3);
\draw[-] (andI1) -- (CFI2);
\draw[-] (andI1) -- (CFI4);
\draw[-] (andI2) -- (CFI2);
\draw[-] (andI2) -- (CFI3);
\draw[-] (andO2) -- (CFI1);
\draw[-] (andO2) -- (CFI4);

\begin{scope}[xshift=2cm]
 \node at (11,-0.2) {$G$ ($c\notin T$)};
\node[knoten] (andO1) at (10.5-0.1,1.6) [label=above:{\small$u_0$}] {};
\node[knoten] (andI1) at (10.5+0.1,1.6) [label=above:{\small$u_1$}]{};
\node[knoten] (andO2) at (11.5-0.1,1.6) [label=above:{\small$v_0$}]{};
\node[knoten] (andI2) at (11.5+0.1,1.6) [label=above:{\small$v_1$}]{};
\node[knoten] (andO3) at (11-0.1,0.4) [label=below:{\small$w_0$}]{};
\node[knoten] (andI3) at (11+0.1,0.4) [label=below:{\small$w_1$}]{};
\node[knoten] (CFI1) at (11-0.3,1) [label=left:{\small$c_0$}] {};
\node[knoten] (CFI2) at (11-0.1,1) [label=below:{\small$c_1$}]{};
\node[knoten] (CFI3) at (11+0.1,1) [label=below:{\small$c_2$}]{};
\node[knoten] (CFI4) at (11+0.3,1) [label=right:{\small$c_3$}]{};
\draw[-] (andO3) -- (CFI1);
\draw[-] (andO3) -- (CFI2);
\draw[-] (andI3) -- (CFI3);
\draw[-] (andI3) -- (CFI4);
\draw[-] (andO1) -- (CFI1);
\draw[-] (andO1) -- (CFI3);
\draw[-] (andI1) -- (CFI2);
\draw[-] (andI1) -- (CFI4);
\draw[-] (andI2) -- (CFI2);
\draw[-] (andI2) -- (CFI3);
\draw[-] (andO2) -- (CFI1);
\draw[-] (andO2) -- (CFI4);
\end{scope}

\begin{scope}[xshift=4cm]
 \node at (11,-0.2) {$\tilde G$};
\node[knoten] (andO1) at (10.5-0.1,1.6) [label=above:{\small$u'_0$}] {};
\node[knoten] (andI1) at (10.5+0.1,1.6) [label=above:{\small$u'_1$}]{};
\node[knoten] (andO2) at (11.5-0.1,1.6) [label=above:{\small$v'_0$}]{};
\node[knoten] (andI2) at (11.5+0.1,1.6) [label=above:{\small$v'_1$}]{};
\node[knoten] (andO3) at (11-0.1,0.4) [label=below:{\small$w'_0$}]{};
\node[knoten] (andI3) at (11+0.1,0.4) [label=below:{\small$w'_1$}]{};
\node[knoten] (CFI1) at (11-0.3,1) [label=left:{\small$c'_0$}] {};
\node[knoten] (CFI2) at (11-0.1,1) [label=below:{\small$c'_1$}]{};
\node[knoten] (CFI3) at (11+0.1,1) [label=below:{\small$c'_2$}]{};
\node[knoten] (CFI4) at (11+0.3,1) [label=right:{\small$c'_3$}]{};
\draw[-] (andO3) -- (CFI1);
\draw[-] (andO3) -- (CFI2);
\draw[-] (andI3) -- (CFI3);
\draw[-] (andI3) -- (CFI4);
\draw[-] (andO1) -- (CFI1);
\draw[-] (andO1) -- (CFI3);
\draw[-] (andI1) -- (CFI2);
\draw[-] (andI1) -- (CFI4);
\draw[-] (andI2) -- (CFI2);
\draw[-] (andI2) -- (CFI3);
\draw[-] (andO2) -- (CFI1);
\draw[-] (andO2) -- (CFI4);
\end{scope}
\end{tikzpicture}
\caption{Cai-Fürer-Immerman Gadets for vertices of degree 3.}
\label{pic:CFI}
\end{figure}

\begin{proof}
  For simplicity we assume that $H$ is 3-regular. Our argument extends
  to graphs $H$ of arbitrary degree. 
  Recall that $G$ and $\tilde G$ contain a uniquely coloured pair of vertices $v_0,v_1$ for every edge $e\in E(H)$.
  Furthermore, for every vertex $t\in V(H)$ incident with edges
  $e_1,e_2e_3$, there are 4 vertices $c^{(t)}_0$, \ldots, $c^{(t)}_3$
  (again of unique colour)
 that are connected to the vertex-pairs
  $0^{(e_i)},0^{(e_i)}$, as shown in Figure \ref{pic:CFI}.
  Now we derive a an unsatisfiable system of linear equations over $\mathbb Z_2$ with a polynomial calculus refutation of rank 1 and degree 2.
  Hence, there is a degree 2 Nullstellensatz refutation, as satisfiability of this system can be refuted using Gaussian elimination over $\mathbb Z_2$.
  For every gadget as depicted in Figure \ref{pic:CFI} we derive 
  \begin{align}
    \label{eq:cfi1}
    x_{u_0u_0'}+x_{v_0v_0'}+x_{w_0w_0'}&=1 &&\text{if $c\in T$,} \\
    x_{u_0u_0'}+x_{v_0v_0'}+x_{w_0w_0'}&=0 &&\text{if $c\notin T$.}
  \end{align}
  By the definition of $\TS(H,T)$, this system is unsatisfiable if $T$ is odd.
  We only prove the former case, the latter is symmetric since (as $x_{w_0w_0'} + x_{w_1w_0'}=1$) it is equivalent to $x_{u_0u_0'}+x_{v_0v_0'}+x_{w_1w_0'}=1$.
  By employing \eqref{eq:comp} for coloured graphs, we directly
  eliminate variables $x_{uv}$, where $u$ and $v$ have different
  colours. The derivation is shown in Figure~\ref{fig:derivation}.
\end{proof}

  \begin{figure}
  \begin{align*}
    1) &&  x_{c_0c_0'} + x_{c_0c_1'} + x_{c_0c_2'} + x_{c_0c_3'} &= 1 &&\text{axiom \eqref{eq:cont2}} \\
    2) &&  x_{w_0w_0'}x_{c_0c_0'} + x_{w_0w_0'}x_{c_0c_1'}\qquad& \\&& + x_{w_0w_0'}x_{c_0c_2'} + x_{w_0w_0'}x_{c_0c_3'} &= x_{w_0w_0'} &&\text{mult. with }x_{w_0w_0'} \\
    3) &&   x_{w_0w_0'}x_{c_0c_2'} &= 0 &&\text{axiom} \\
    4) &&   x_{w_0w_0'}x_{c_0c_3'} &= 0 &&\text{axiom} \\
    5) &&  x_{w_0w_0'}x_{c_0c_0'} + x_{w_0w_0'}x_{c_0c_1'}  &= x_{w_0w_0'} &&(2) - (3) - (4) \\
    6) &&   x_{w_0w_0'}+x_{w_0w_1'} &= 1 &&\text{axiom \eqref{eq:cont2}}\\
    7) &&   x_{c_0c_0'}x_{w_0w_0'}+x_{c_0c_0'}x_{w_0w_1'} &= x_{c_0c_0'} &&\text{mult. with }x_{c_0c_0'} \\
    8) &&   x_{c_0c_0'}x_{w_0w_1'} &= 0 &&\text{axiom} \\
    9) &&   x_{c_0c_0'}x_{w_0w_0'} &= x_{c_0c_0'} &&(7)-(8) \\
    10) &&   x_{c_0c_1'}x_{w_0w_0'}+x_{c_0c_1'}x_{w_0w_1'} &= x_{c_0c_1'} &&\text{mult. (6) with }x_{c_0c_1'} \\
    11) &&   x_{c_0c_1'}x_{w_0w_1'} &= 0 &&\text{axiom} \\
    12) &&   x_{c_0c_1'}x_{w_0w_0'} &= x_{c_0c_1'} &&(10)-(11) \\
    13) &&  x_{c_0c_0'} + x_{c_0c_1'}  &= x_{w_0w_0'} &&(5) - (9) - (12) \\
    14) &&  x_{c_0c_0'} + x_{c_0c_2'}  &= x_{u_0u_0'} &&\text{analogous to (1)\ldots(13) with $u$} \\
    15) &&  x_{c_0c_0'} + x_{c_0c_3'}  &= x_{v_0v_0'} &&\text{analogous to (1)\ldots(13) with $v$} \\
    16) &&  x_{u_0u_0'}+x_{v_0v_0'}+x_{w_0w_0'}  &= 2x_{c_0c_0'}+1 &&(1)-(13)-(14)-(15) \\
    16) &&  x_{u_0u_0'}+x_{v_0v_0'}+x_{w_0w_0'}  &= 1 &&\text{over $\mathbb Z_2$} \qed\\
  \end{align*}
  \caption{Derivation of \eqref{eq:cfi1}}\label{fig:derivation}
  \end{figure}

  Thus, to prove lower bounds for algebraic proof systems over $\FF_2$
  we need new techniques. Our final theorem, which even derives lower
  bound over $\mathbb Z$, is a first step.

\begin{theorem}\label{theo:int}
There non-isomorphic graphs $G$, $H$ such that
$\operatorname{MLIN}(\GIeq{G,H}^2)$ has a solution over $\mathbb
Z$. 
\end{theorem}

\begin{corollary}
  There non-isomorphic graphs $G$, $H$ such that $\GIeq{G,H}$ has no degree-2 Nullstellensatz refutation over $\mathbb F_q$ for any prime $q$.
\end{corollary}

  For two graphs $G$ and $H$, consider a colouring of the vertices of both graphs. 
  We call such a colouring \emph{suitable}, if 
\begin{enumerate}
  \item for every colour, the number of vertices of this colour is the same in both graphs,
  \item the colour classes form independent sets in both graphs, and
  \item the induced subgraph on two distinct colour classes has either no edge or is a matching between these two colour classes (and the shape is the same in both graphs).
\end{enumerate}
  The \emph{index} of a colouring be the product of all different colour class sizes.

\begin{lemma}\label{lem:coprime_colouring}
  Given $G$ and $H$. If there are two suitable colourings with co-prime index, then  $\operatorname{MLIN}(\GIeq{G,H}^2)$ has a solution over $\mathbb Z$.
\end{lemma}

\begin{proof}
  For a colouring of index $c$, we define an integer assignment $\alpha$ to the variables of $\operatorname{MLIN}(\GIeq{G,H}^2)$ as follows: 
  \begin{align}
    \alpha(X_{vw}) = 
    \begin{cases}
      0 \text{, if $v$ and $w$ have different colours,} \\
      c/\ell \text{, if $\ell$ is the size of the colour class of $v$ and $w$.} 
    \end{cases}
  \end{align}
  For the variables $X_{vw,v'w'}$ we set $\alpha(X_{vw,v'w'})=0$, if the colours of $v$ and $w$ or the colours of $v'$ and $w'$ differ. 
  Otherwise, let $C_1$ be the colour class of $v$ and $w$ and $C_2$ be the colour class of $v'$ and $w'$.
  For a colour class $C$, denote by $C^G$ and $C^H$ the vertices of colour $C$ in $G$ and $H$, respectively.
  Now we fix two colour classes $C_1$ and $C_2$. 
  Let $C_1^G = \{v_1,\ldots,v_{\ell}\}$, $C_2^G = \{v'_1,\ldots,v'_{m}\}$, $C_1^H = \{w_1,\ldots,w_{\ell}\}$ and $C_2^H = \{w'_1,\ldots,w'_{m}\}$.
  If $C_1=C_2$, it holds that $\ell=m$, $v_i=v'_i$ and $w_i=w'_i$.
  In the case that $C_1$ and $C_2$ are distinct and the edge set forms a matching, we assume w.l.o.g. that the edges are $\{v_i,v'_i\}\in E(G)$ and $\{w_i,w'_i\}\in E(H)$ for $1\leq i\leq \ell=m$. 
  In both cases let
  \begin{align}
    \alpha(X_{\{v_iw_j,v'_aw'_b\}}) = \begin{cases}
      c/\ell \text{, if $i-a=j-b$,} 
      0 \text{, otherwise.} \\
    \end{cases}
  \end{align}  
  In the remaining case that $C^1$ and $C^2$ are distinct and there are no edges between the colour classes, we let $\alpha(X_{\{vw,v'w'\}})= c/(\ell m)$.
  This mapping ensures, that $\alpha(X_{\pi})=0$, if $\pi$ is no local isomorphism.
  Furthermore, for the linearised 2-dimensional version of \eqref{eq:cont1} and \eqref{eq:cont2} we have 
  \begin{align}
    \alpha(X_{\{vw\}}) &= \sum_{v'\in V(G)} \alpha(X_{\{vw,v'w'\}}) \text{, for all $w'\in V(H)$, and} \\
    \alpha(X_{\{vw\}}) &= \sum_{w'\in V(H)} \alpha(X_{\{vw,v'w'\}}) \text{, for all $v'\in V(G)$.} 
  \end{align}
  However, the linearised version \eqref{eq:cont1} and \eqref{eq:cont2} itself we have
  \begin{align}
    \sum_{v\in V(G)} \alpha(X_{\{vw\}}) &= c \text{, for all $w\in V(H)$, and}\\
    \sum_{w\in V(H)} \alpha(X_{\{vw\}}) &= c \text{, for all $v\in V(G)$.} 
  \end{align}

  Hence, $\alpha$ is only a satisfying assignment if $c=1$ (and in this case $G$ and $H$ are isomorphic).
  However, let $\beta$ be the assignment constructed the same way out of a second suitable colouring with index $b$.
  If $b$ and $c$ are co-prime then there are integers $s$ and $t$ such that $1=sb+tc$.
  Thus, the assignment $\gamma$ defined by 
  $$
  \gamma(X_\pi) = s\beta(X_\pi) + t\alpha(X_\pi)
  $$
  satisfies all equations from $\operatorname{MLIN}(\GIeq{G,H}^2)$.
\end{proof}

\begin{tikzpicture}
  [scale=.6, knoten/.style={circle,draw=black,fill=black,
  inner  sep=0pt,minimum  size=1.5mm}]

  \node[knoten,fill=green] (1) at (1.5,2.6) {};
  \node[knoten,fill=green] (2) at (1.5,3.1) {};
  \node[knoten,fill=blue] (3) at (0,0) {};
  \node[knoten,fill=blue] (4) at (0,.5) {};
  \node[knoten,fill=red] (5) at (3,0) {};
  \node[knoten,fill=red] (6) at (3,.5) {};

  \node[knoten,fill=green] (a) at (1.5,2.6-.5) {};
  \node[knoten,fill=blue] (b) at (0,0-.5) {};
  \node[knoten,fill=red] (c) at (3,0-.5) {};
 \draw[-] (4) -- (2);
 \draw[-] (3) -- (1);

 \draw[-] (2) -- (6);
 \draw[-] (1) -- (5);

 \draw[-] (4) -- (6);
 \draw[-] (3) -- (5);

 \draw[-] (a) -- (b);
 \draw[-] (b) -- (c);
 \draw[-] (c) -- (a);

\node at (1.5,-1) {$G$};

\begin{scope}[xshift=5cm]

  \node[knoten,fill=green] (1) at (1.5,2.6) {};
  \node[knoten,fill=green] (2) at (1.5,3.1) {};
  \node[knoten,fill=blue] (3) at (0,0) {};
  \node[knoten,fill=blue] (4) at (0,.5) {};
  \node[knoten,fill=red] (5) at (3,0) {};
  \node[knoten,fill=red] (6) at (3,.5) {};

  \node[knoten,fill=green] (a) at (1.5,2.6-.5) {};
  \node[knoten,fill=blue] (b) at (0,0-.5) {};
  \node[knoten,fill=red] (c) at (3,0-.5) {};

 \draw[-] (4) -- (2);
 \draw[-] (3) -- (1);

 \draw[-] (2) -- (6);
 \draw[-] (1) -- (5);

 \draw[-] (4) -- (5);
 \draw[-] (3) -- (6);

 \draw[-] (a) -- (b);
 \draw[-] (b) -- (c);
 \draw[-] (c) -- (a);

\node at (1.5,-1) {$H$};
\end{scope}

\node at (4,-2) {suitable colouring of index 3};

\begin{scope}[xshift=11cm]

  \node[knoten,fill=green] (1) at (1.5,2.6) {};
  \node[knoten,fill=green] (2) at (1.5,3.1) {};
  \node[knoten,fill=blue] (3) at (0,0) {};
  \node[knoten,fill=blue] (4) at (0,.5) {};
  \node[knoten,fill=red] (5) at (3,0) {};
  \node[knoten,fill=red] (6) at (3,.5) {};

  \node[knoten,fill=white] (a) at (1.5,2.6-.5) {};
  \node[knoten,fill=black] (b) at (0,0-.5) {};
  \node[knoten,fill=yellow] (c) at (3,0-.5) {};
 \draw[-] (4) -- (2);
 \draw[-] (3) -- (1);

 \draw[-] (2) -- (6);
 \draw[-] (1) -- (5);

 \draw[-] (4) -- (6);
 \draw[-] (3) -- (5);

 \draw[-] (a) -- (b);
 \draw[-] (b) -- (c);
 \draw[-] (c) -- (a);
\node at (1.5,-1) {$G$};

\begin{scope}[xshift=5cm]

  \node[knoten,fill=green] (1) at (1.5,2.6) {};
  \node[knoten,fill=green] (2) at (1.5,3.1) {};
  \node[knoten,fill=blue] (3) at (0,0) {};
  \node[knoten,fill=blue] (4) at (0,.5) {};
  \node[knoten,fill=red] (5) at (3,0) {};
  \node[knoten,fill=red] (6) at (3,.5) {};

  \node[knoten,fill=white] (a) at (1.5,2.6-.5) {};
  \node[knoten,fill=black] (b) at (0,0-.5) {};
  \node[knoten,fill=yellow] (c) at (3,0-.5) {};

 \draw[-] (4) -- (2);
 \draw[-] (3) -- (1);

 \draw[-] (2) -- (6);
 \draw[-] (1) -- (5);

 \draw[-] (4) -- (5);
 \draw[-] (3) -- (6);

 \draw[-] (a) -- (b);
 \draw[-] (b) -- (c);
 \draw[-] (c) -- (a);
\node at (1.5,-1) {$H$};
\end{scope}

\node at (4,-2) {suitable colouring of index 2};
  
\end{scope}

\end{tikzpicture}

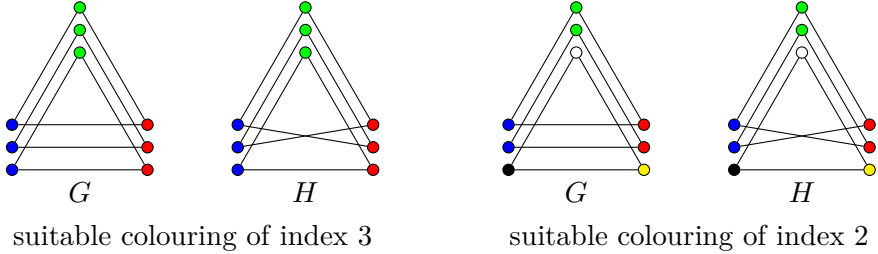
\begin{proof}[Proof of Theorem~\ref{theo:int}]
    Let $G$ be the disjoint union of three triangles and $H$ be the disjoint union of a triangle and a 6-cycle (with vertices $w_0,\ldots,w_5$ and edges $\{w_i,w_{i+1 (\text{mod } 6)}\}$).
  First, we colour the three vertices of every triangle with the three colours $C_0,C_1,C_2$. 
  In the 6-cycle, we colour $w_i$ and $w_{i+3}$ with $C_i$ (for $i=0,1,2$).
  As every colour contains exactly three elements, this is a suitable colouring of index 3. 
  We now change the colouring and assign new colours $C_3$, $C_4$, $C_5$ to the triangle in $H$ and an arbitrary triangle in $G$.
  Again, it is not hard to see, that this colouring is suitable and has index 2.
  Thus, by Lemma~\ref{lem:coprime_colouring}, $\operatorname{MLIN}(\GIeq{G,H}^2)$ has a solution over $\mathbb Z$.
\end{proof}

\section{Concluding Remarks}
Employing results and techniques from propositional proof complexity,
we prove strong lower bounds for algebraic algorithms for graph
isomorphism testing, which show that these algorithm are not much
stronger than known algorithms such as the Weisfeiler-Lehman
algorithm. 

Our results hold over all fields except---surprisingly---fields of
characteristic $2$. For fields of characteristic 2, and also for the ring of integers,
we only have very weak lower bounds. It remains an challenging open
problem to improve these.

\subsection*{Acknowledgements}
We thank Anuj Dawar and Erkal Selman for many inspiring discussions in
the initial phase of this project.

\end{document}